\documentclass[12pt]{article}

\usepackage{amsmath}
\usepackage{amssymb}
\usepackage{amsthm}
\usepackage{tikz}
\usepackage{hyperref}
\usepackage{enumitem}

\newtheorem{thm}{Theorem}[section]

\newtheorem{lem}[thm]{Lemma}

\newtheorem{assumption}[thm]{Assumption}

\newtheorem{definition}[thm]{Definition}
\newenvironment{de}{\begin{definition}\rm}{\end{definition}}
\newtheorem{example}[thm]{Example}
\newenvironment{exmp}{\begin{example}\rm}{\end{example}}
\newtheorem{remark}[thm]{Remark}
\newenvironment{rem}{\begin{remark}\rm}{\end{remark}}

\def\eps{\varepsilon}

\title{A Randomized Approach to the Capacity of Finite-State Channels}

\author{\begin{tabular}{cc}
Guangyue Han\\
The University of Hong Kong\\
{\em email:} ghan@hku.hk\\
\end{tabular}}

\date{{\normalsize \today}}

\begin{document}\maketitle\thispagestyle{empty}

\begin{abstract}
Inspired by the ideas from the field of stochastic approximation, we propose a randomized algorithm to compute the capacity of a finite-state channel with a Markovian input. When the mutual information rate of the channel is concave with respect to the chosen parameterization, we show that the proposed algorithm will almost surely converge to the capacity of the channel and derive the rate of convergence. We also discuss the convergence behavior of the algorithm without the concavity assumption.
\end{abstract}

\section{Introduction} \label{introduction}

Discrete-time finite-state channels are a broad class of channels which have attracted plenty of interest in information theory; prominent examples of such channels include partial response channels~\cite{pr00, th87}, Gilbert-Elliott channels~\cite{mu89, go96} and noisy input-restricted channels~\cite{ZehaviWolf88}, which are widely used in a variety of real-life applications, including magnetic and optical recording~\cite{mrs98}, communications over band-limited channels with inter-symbol interference~\cite{fo72}. The computation of the capacity of a finite-state channel is notoriously difficult and has been open for decades. For a discrete memoryless channel with a discrete memoryless source at its input, the classical Blahut-Arimoto algorithm (BAA)~\cite{ar72, bl72} can effectively compute the channel capacity, however, for almost all nontrivial finite-state channels, little is known about the channel capacity other than some bounds (see, e.g.,~\cite{ZehaviWolf88},~\cite{ShamaiKofman90},~\cite{Arnold_etal06} and references therein), which are numerically computed using Monto Carlo approaches. The methods in these work are believed to produce fairly precise numerical approximations of the capacity of certain classes of finite-state channels, however there are no general proofs to justify such beliefs.

Recently, Vontobel {\em et al.} have proposed a generalized Blahut-Arimoto algorithm (GBAA) \cite{pa04} to maximize the mutual information rate of a finite-state machine channel with a finite-state machine source at its input. This interesting algorithm has attracted a great deal of attention due to the observations that it fairly precisely approximates the channel capacity for a number of practical channels. (Notably, some results that were derived in the context of the GBAA have proven to be useful for analyzing the Bethe entropy function of some graphical models that appear in the context of low-density parity-check codes~\cite{vo10} and for approximately computing the permanent of a non-negative matrix~\cite{vo12}.) For a finite-state channel, let $X$ denote the input Markov process and $Y$ its corresponding output process, which, by definition, is a {\em hidden Markov process}~\cite{bo10}. In contrast to the BAA, the convergence of the GBAA depends on the extra assumption that  $I(X; Y)$ and $H(X|Y)$ are both concave with respect to a chosen parameterization, which has been posed as Conjecture $74$ in~\cite{pa04}. Example~\ref{not-concave-example}, however, shows that the concavity conjecture is not true in general; for other examples showing $I(X; Y)$ and $H(X|Y)$ fail to be concave, see~\cite{LiHan}.

One of the hurdles encountered in computing the finite-state channel capacity is the problem of optimizing $H(Y)$, which naturally occurs in the formula of the capacity of a broad class of finite-state channels. More specifically, there has long been a lack of understanding on the following two issues:
\begin{itemize}
\item[(I)]  How to effectively compute the entropy rate of hidden Marov processes?
\item[(II)] How does the entropy rate of hidden Markov processes vary as the underlying Markov processes and the channels vary?
\end{itemize}

As elaborated below, recently, these two issues have been partially addressed by the information theory community.

\textbf{Related work on (I).} It is well known that $H(X)$ has a simple analytic formula; in stark contrast, there is no simple and explicit formula of $H(Y)$ for most non-degenerate channels ever since hidden Markov processes (or, more precisely, hidden Markov models) were formulated more than half a century ago. Here, we remark that Blackwell~\cite{bl57} showed that $H(Y)$ can be written as an integral of an explicit function on a simplex with respect to the Blackwell Measure. However, the Blackwell measure seems to be rather complicated for effective computation of $H(Y)$. Since 2000, there has been a rebirth of interest in computing and estimating $H(Y)$ in a variety of scenarios: the Blackwell measure has been used to bound $H(Y)$~\cite{or03}, a variation on the classical Birch bounds~\cite{bi62} can be found in~\cite{eg04} and a new numerical approximation of $H(Y)$ has been proposed in~\cite{gu09}. Generalizing Blackwell's idea, an integral formula for the derivatives of $H(Y)$ has been derived in~\cite{pf10}.

The celebrated Shannon-McMillan-Breiman theorem states that the $n$-th order {\em sample entropy} $-\log p(Y_1^n)/n$ converges to $H(Y)$ almost surely. Based on this, efficient Monte Carlo methods for approximating $H(Y)$ were proposed independently by Arnold and Loeliger~\cite{ar01}, Pfister, Soriaga and Siegel~\cite{pf01}, Sharma and Singh~\cite{sh01}. However, more quantitative description of the convergence behavior of the proposed methods, such as rate of convergence, asymptotic normality and so on, are lacking in these work. Recently, a central limit theorem (CLT)~\cite{pf03} for the sample entropy has been derived as a corollary of a CLT for the top Lyapunov exponent of a product of random matrices; a functional CLT has also been established in~\cite{ho03}. To some extent, these two CLTs suggested that the Monte Carlo methods are ``accurate'' in terms of approximating $H(Y)$. There are also other related work in different contexts from outside the information theory community, such as~\cite{ko98, hv04, Haydn09}.

Recently, we have obtained~\cite{han11} a number of limit theorems for the sample entropy of $Y$. These limit theorems can be viewed as further refinements of the Shannon-McMillian-Breiman theorem, which is the backbone of information theory. More specifically, Theorem $1.2$ in~\cite{han11} is a CLT with an error-estimate, which can be used to characterize the rate of convergence of the Monte Carlo methods in~\cite{ar01, pf01, sh01}, and Theorem $1.5$ in~\cite{han11} is a large deviation result, which gives a sub-exponential decaying upper bound on the probability of the sample entropy $-\log p(Y_1^n)/n$ deviating from $H(Y)$. Among many other applications, such as deriving non-asymptotic coding theorems~\cite{yangenhui}, these theorems positively confirmed the effectiveness of using the Shannon-McMillan-Breiman theorem to approximate $H(Y)$.

\textbf{Related work on (II).} The behavior of $H(Y)$ (as a function of the underlying Markov chain and the channel) is of significance in a number of scientific disciplines; particularly in information theory, it is of great importance for computing/estimating the capacity of finite-state channels.
However, some of the basic problems, such as smoothness (or even differentiability) of $H(Y)$, have long remained unknown. Recently, asymptotical behavior of $H(Y)$ has been studied in~\cite{ar94a,jss08,or03,or04, zu04,zu05,na05,an10,pf10}. Particularly in~\cite{zu04}, for a special type of hidden Markov chain $Y$, the Taylor series expansion of $H(Y)$ is given under the assumption that $H(Y)$ is analytic. Under mild assumptions, analyticity of $H(Y)$ has been established in~\cite{gm05}; see also related work in~\cite{bo10, zu04, zu05, al08, gu09, pf10} and references therein. The framework in~\cite{gm05} has been generalized to continuous-state settings and further provides useful tools and techniques for our subsequent work, such as derivatives~\cite{hm06b}, asymptotics~\cite{hm09}, concavity~\cite{hm09c} of $H(Y)$.

Equipped with ideas and techniques from the above-mentioned work on (I) and (II), we are more prepared to make further progress towards the computation of the channel capacity. In particular, the ideas and techniques in~\cite{han11} and~\cite{gm05} are vital to this paper. Roughly speaking, \cite{gm05} proves that the entropy rate of hidden Markov chains is a ``nicely behaved'' function; and~\cite{han11} confirms that it can be ``well-approximated'' using Monte Carlo simulations. The simulator of the derivative of $I(X; Y)$ as specified in Section~\ref{simulator}, which is crucial to this work, is an ``offspring'' of the two schools of thoughts in~\cite{gm05} and~\cite{han11}.

Stochastic approximation methods refer to a family of recursive stochastic algorithms, aiming to find zeroes or extrema of functions whose values can only be estimated via noisy observations. The extensive literature on stochastic approximation has grown up around two prototyipcal algorithms, the Robbins-Monro algorithm and the Kiefer-Wolfowitz algorithm, mainly concerning the convergence analysis on these two algorithms and their variants; we refer the reader to~\cite{ku03} for an exposition to the vast literature on stochastic approximation.

Inspired by the ideas in stochastic approximation, we propose a randomized algorithm to compute the capacity of a class of finite-state channels with input Markov processes supported on some mixing finite-type mixing constraint. Bearing the same spirit as the Robbins-Monro algorithm and the Kiefer-Wolfowitz algorithm, the proposed algorithm, in many subtle respects, differs from both of them. The main task of this paper is to conduct a convergence analysis of the proposed algorithm, which employs some established ideas and techniques from the field of stochastic approximation. In particular, the proofs in Section~\ref{without-concavity} are largely inspired by \cite{ta10}, which has credited origins of some of its techniques to earlier work, such as~\cite{be90, ku03, lj99}. However, neither the results nor the proofs in~\cite{ta10} or any of previous work imply our results; as a matter of fact, considerable amount of simplification and adaptation of the techniques in~\cite{ta10} have been incorporated into this work.

Although described in different languages, our settings are essentially the same as in~\cite{pa04}. On the other hand, as opposed to the GBAA, the concavity of $I(X; Y)$ alone is already sufficient to guarantee the convergence of our algorithm. Here, let us note that that for certain classes of channels (see Example~\ref{not-concave-example}), $I(X; Y)$ is indeed concave with respect to certain parameterization, whereas $H(X|Y)$ fails to be concave with respect to the same parameterization.

Characterizing the maximal rate at which the information can be transmitted through a given channel, the capacity is the most fundamental notion in information theory. The capacity achieving distribution will further provide us insightful guidance towards designing coding schemes that actually achieve the promised capacity. Apparently, such an algorithm would be of fundamental significance to both information theoretic research and practical applications to tele-communications and data storage.

The organization of the paper is as follows. We first describe our channel model in greater detail in Section~\ref{channel-model} and we then present our algorithm  in Section~\ref{the-algorithm}. In Section~\ref{simulator}, we propose a simulator for the derivative of $I(X; Y)$ and discuss its convergence behavior. The convergence of the algorithm is established in~\ref{f-convergence}, while the rate of convergence of the algorithm with and without concavity conditions are derived in Sections~\ref{with-concavity} and~\ref{without-concavity}, respectively. In Section~\ref{Memoryless-Channels}, we discuss the capacity achieving distribution of a special class of finite-state channels.

\section{Channel Model} \label{channel-model}
In this section, we specify the channel model considered in this paper in greater detail, which is essentially the same as the one considered in~\cite{pa04}.

Let $\mathcal{X}$ be a finite alphabet and let
$$
\mathcal{X}^2 = \{(i, j): i, j \in \mathcal{X}\}.
$$
Let $\Pi$ denote the set of all stationary irreducible first-order Markov chain over the alphabet $\mathcal{X}$. For a given subset $F \subset \mathcal{X}^2$, define
$$
\Pi_F=\{X \in \Pi: X_{i, j}=0, ~~ (i, j) \in F\},
$$
where we have identified an irreducible first-order Markov chain with its transition probability matrix. Furthermore, for any $\epsilon > 0$, define
$$
\Pi_{F, \epsilon}=\{X \in \Pi_F: X_{i, j} \geq \epsilon, ~~ (i, j) \not \in F\}.
$$
Obviously, if some $X \in \Pi_{F, \epsilon}$ is primitive (namely, irreducible and aperiodic), then any $X' \in \Pi_{F, \epsilon}$ is primitive; in this case, we say $F$ is a {\em mixing} finite-type constraint. Here, let us note that a mixing finite-type constraint can be defined in a much more general context; see~\cite{lm95}.

The motivation for consideration of finite-type constraints mainly comes from magnetic recording, where input sequences are required to satisfy certain mixing finite-type constraints in order to eliminate the most damaging error events~\cite{mrs98}. The most well known example is the so-called $(d,k)$-RLL constraint $\mathcal{S}(d,k)$ over the alphabet $\{0, 1\}$, which forbids any sequence with fewer than $d$ or more than $k$ consecutive zeros in between two
successive $1$'s.

In this paper, we are concerned with a discrete-time finite-state channel with some input constraint. Let $X, Y, S$ denote the channel input, output and state processes over finite alphabets $\mathcal{X}, \mathcal{Y}$ and $\mathcal{S}$, respectively. Assume that
\begin{enumerate}
\item[(\ref{channel-model}.a)] For some mixing finite-type constraint $F \subset \mathcal{X}^2$ and some $\epsilon > 0$, $X \in \Pi_{F, \epsilon}$.
\item[(\ref{channel-model}.b)] $(X, S)$ is a first-order stationary Markov chain whose transition probabilities satisfy
$$
p(x_n, s_n|x_{n-1}, s_{n-1})=p(x_n|x_{n-1}) p(s_n|x_n, s_{n-1}),
$$
where $p(s_n|x_n, s_{n-1}) > 0$ for any $s_{n-1}, s_{n}, x_n$.
\item[(\ref{channel-model}.c)] the channel is stationary, and the channel transition probabilities satisfy
$$
p(y_n, s_n|x_n, s_{n-1})=p(s_n|x_n, s_{n-1}) p(y_n|x_n, s_n).
$$
\end{enumerate}
The capacity of the above channel is defined as
$$
C_{F}=\sup I(X; Y)= \sup \lim_{n \to \infty} I_n(X; Y),
$$
where the supremum is over all $X$ satisfying (\thesection.a) and
$$
I_n(X; Y) \triangleq \frac{H(X_1^n)+H(Y_1^n)-H(X_1^n, Y_1^n)}{n}.
$$
The fact that $Y$ and $(X, Y)$ are both hidden Markov processes makes it apparent that solutions to (I) and (II) are essential for computing $C_F$.

Assume that $\Pi_{F, \epsilon}$ is analytically parameterized by $\theta \in \Theta \subset \mathbb{R}^d$, $d \geq 1$, where $\Theta$ denote the entire parameter space. Then, naturally, $X=X(\theta)$ and $Y=Y(\theta)$ are also analytically parameterized by $\theta$. Under this parameterization, we would like to find $\theta^* \in \Theta$ such that $X(\theta^*)$ maximizes $ I(X(\theta); Y(\theta))$.

\begin{rem}
One natural goal is to find $X \in \Pi_{F}$ to maximize $I(X; Y)$. However, in this paper, we will restrict our attention to $\Pi_{F, \epsilon}$ for a given $\epsilon > 0$; such restriction will be justified in Section~\ref{Memoryless-Channels}.
\end{rem}

\section{The Algorithm} \label{the-algorithm}

For a given $1/2 < a < 1$, choose the so-called step sizes
$$
a_n=\frac{1}{n^a}, \qquad n=1, 2, \cdots;
$$
apparently, $\{a_n\}$ satisfies
$$
\sum_{n=0}^{\infty} a_n = \infty, \quad \sum_{n=0}^{\infty} a_n^2 < \infty,
$$
which are the typical conditions imposed on step sizes in a generic stochastic approximation method. Letting $A_n$ denote the event ``$\theta_n + a_n g_{n^b}(\theta_n) \not \in \Theta$'', we propose to find $\theta^*$ through the following recursive procedure:
\begin{equation}  \label{theta-iteration}
\theta_{n+1}=\begin{cases}
\theta_n, \mbox{ if } A_n \mbox{ occurs, } \\
\theta_n + a_n g_{n^b}(\theta_n), \mbox{ otherwise; }
\end{cases}
\end{equation}
here $b > 0$, the initial $\theta_0$ is randomly selected from $\Theta$, and $g_{n^b}(\theta)$ is a to-be-specified simulator (see Section~\ref{simulator}) for $I'(X(\theta); Y(\theta))$, where the derivative is taken with respect to $\theta$. Throughout the paper, we assume that
\begin{equation}  \label{global-conditions}
0 < \beta < \alpha < 1/3, ~~ 2a+b-3 b \beta > 1;
\end{equation}
here, $\alpha, \beta$ are some ``hidden'' parameters involved in the definition of $g_{n^b}(\theta)$, which will be defined in Section~\ref{simulator}.

\section{A Simulator of $I'(X; Y)$} \label{simulator}

As stated in Section~\ref{introduction}, albeit rather difficult to compute analytically, $I_n(X; Y)$ can be well-approximated via Monte Carlo simulations.  In this section, we propose a simulator for $I'(X; Y)$. Needlessly to say, an effective simulator guaranteeing an ``accurate'' approximation to $I'(X; Y)$ is crucial to our algorithm. To some extent, our simulator is inspired by the Bernstein's blocking method~\cite{be26}, which is a well-established tool in proving limit theorems for mixing sequences; see, e.g.,~\cite{br07}.

Now, consider a stationary stochastic process $Z=Z_{-\infty}^{\infty}$ satisfying the following assumptions:
\begin{enumerate}
\item[(\ref{simulator}.a)] There exist $C', C'' > 0$ such that for all $z_{-n}^0$,
$$
C' \leq p(z_0|z_{n}^{-1}) \leq C''.
$$

\item[(\ref{simulator}.b)] There exist $C > 0$, $0 < \lambda < 1$ such that for all $n$,
$$
\psi_Z(n) \triangleq \sup_{U \in \mathcal{B}(Z_{\infty}^{-n}), V \in \mathcal{B}(Z_{0}^{\infty}), P(U) > 0, P(V) > 0} |P(V|U)-P(V)|/P(V) \leq C \lambda^n,
$$
where $\mathcal{B}(Z_i^j)$ denotes the $\sigma$-field generated by $\{Z_k: k=i, i+1, \cdots, j\}$.

\item[(\ref{simulator}.c)] There exist $C > 0$, $0 < \rho < 1$ such that for any two $z_{-m}^0, \hat{z}_{-\hat{m}}^0$ with $z_{-n}^0=\hat{z}_{-n}^0$ (here $m, \hat{m} \geq n \geq 0$),
$$
|p(z_0|z_{-m}^{-1}) - p(\hat{z}_0|\hat{z}_{-\hat{m}}^{-1}) | \leq C \rho^n.
$$
\end{enumerate}

\begin{rem}  \label{Z-Created}
Conditions (\ref{simulator}.a)-(\ref{simulator}.c) are the same ones used in Section $2$ of~\cite{han11}, which are essential for establishing the main results in~\cite{han11}. As observed in~\cite{han11}, Condition (\ref{channel-model}.a) implies that $Y$ and $(X, Y)$ both satisfy Conditions (\ref{simulator}.a)-(\ref{simulator}.c).
\end{rem}

Now, for $0 < \beta < \alpha < 1/3$, define
$$
q=q(n) \triangleq n^{\beta}, ~~ p=p(n) \triangleq n^{\alpha}, ~~ k=k(n) \triangleq n/(n^{\alpha}+n^{\beta}).
$$
For any $j$ with $iq+(i-1)p+1 \leq j \leq iq+ip$, define
$$
\hspace{-1cm} W_j = W_j(Z_{j-\lfloor q/2 \rfloor}^j) \triangleq -\left(\frac{p'(Z_{j - \lfloor q/2 \rfloor})}{p(Z_{j - \lfloor q/2 \rfloor})}+ \frac{p'(Z_{j - \lfloor q/2 \rfloor+1}|Z_{j - \lfloor q/2 \rfloor})}{p(Z_{j - \lfloor q/2 \rfloor+1}|Z_{j - \lfloor q/2 \rfloor})}
+\cdots+ \frac{p'(Z_{j}|Z_{j- \lfloor q/2 \rfloor}^{j-1})}{p(Z_{j}|Z_{j- \lfloor q/2 \rfloor}^{j-1})} \right) \log p(Z_{j}|Z_{j-\lfloor q/2 \rfloor}^{j-1}),
$$
and furthermore
$$
\zeta_i \triangleq W_{i q+ (i-1)p +1}+ \cdots+ W_{i q+i p}, ~~ S_n \triangleq \sum_{i=1}^{k(n)} \zeta_i.
$$

Now, we are ready to define our simulator for $I'(X; Y)$.
\begin{de}
$$
g_n= g_n(X_1^n, Y_1^n) \triangleq H'(X_2|X_1) + S_n(Y_1^n)/(k p) - S_n(X_1^n, Y_1^n)/(k p).
$$
\end{de}

The following lemma, whose proof is somewhat similar to that of Lemma $3.3$ in~\cite{han11}, gives an estimate of the variance of $S_n$.
\cite{han11}.
\begin{lem} \label{order-of-mements}
For $Z$ satisfying Conditions (\ref{simulator}.a), (\ref{simulator}.b) and (\ref{simulator}.c),
$$
E[(S_n-E[S_n])^2] = O(k p q^3).
$$

\end{lem}

\begin{proof}

As in~\cite{han11}, using Condition (\ref{simulator}.a), (\ref{simulator}.b), we can deduce that for some $0 < \lambda < 1$,
$$
E[(S_n-E[S_n])^2]=E[(\sum_{i=1}^{k} \zeta_i-\sum_{i=1}^{k} E[\zeta_i])^{2}]=k E[(\zeta_i-E[\zeta_i])^{2}]+ O(k^2 \lambda^{q/2}).
$$
So, to prove the lemma, it suffices to prove that for any $i \in \mathbb{N}$,
$$
E[(\zeta_i-E[\zeta_i])^{2}]=O(p q^3).
$$
Note that
\begin{equation} \label{expression-1}
E[(\zeta_i-E[\zeta_i])^{2}]= E[(\sum_{i=1}^k W_i-E[W_i])^{2}]
=\sum_{i, j=1}^k E[(W_i-E[W_i])(W_j-E[W_j])].
\end{equation}
It is apparent that when $|j - i| \leq \lfloor q/2 \rfloor$,
\begin{equation}  \label{expression-2}
E[(W_i-E[W_i])(W_j-E[W_j])]=O(q^2),
\end{equation}
and one verifies, using Condition (\ref{simulator}.a), (\ref{simulator}.b), that when $|j - i| > \lfloor q/2 \rfloor$,
\begin{equation}  \label{expression-3}
E[(W_i-E[W_i])(W_j-E[W_j])]= O(q^2 \lambda^{|j-i|-\lfloor q/2 \rfloor}).
\end{equation}
Combining (\ref{expression-1}), (\ref{expression-2}) and (\ref{expression-3}), we then have
\begin{align*}
E[(\zeta_i-E[\zeta_i])^{2}] & =(\sum_{|j - i| \leq \lfloor q/2 \rfloor}+ \sum_{|j - i| > \lfloor q/2 \rfloor})E[(W_i-E[W_i])(W_j-E[W_j])]\\
&=O(p q^3).
\end{align*}
The proof is then complete.

\end{proof}

The following three theorems characterise the performances of our simulator from different perspectives.

Using similar techniques as in the proof of Theorem $1.1$ in~\cite{gm05}, the first theorem shows that on average, our simulator sub-exponentially converges to $I'(X; Y)$.
\begin{thm}  \label{expectation-convergence}
For some $0 < \rho_0 < 1$, we have
$$
E[g_n(X_1^n, Y_1^n)]-I'(X; Y) = O(\rho_0^{\lfloor q/2 \rfloor}).
$$
\end{thm}

\begin{proof}

Notice that for the Markov chain $X$, we have
$$
H(X)=H(X_2|X_1).
$$
So, by Remark~\ref{Z-Created}, it suffices to prove that for any $Z$ satisfying Conditions (\ref{simulator}.a)-(\ref{simulator}.c), we have
$$
\frac{S_n}{k p}-H'(Z) = O(\rho_1^{\lfloor q/2 \rfloor}),
$$
for some $0 < \rho_1 < 1$.

Note that for any $j$ with $iq+(i-1)p+1 \leq j \leq iq+ip$, we have
\begin{align*}
\hspace{-1cm} E[W_j] & =-\sum_{z_{j-\lfloor q/2 \rfloor}^j} p(z_{j-\lfloor q/2 \rfloor}^j) \left(\frac{p'(z_{j - \lfloor q/2 \rfloor})}{p(z_{j - \lfloor q/2 \rfloor})}+ \frac{p'(z_{j - \lfloor q/2 \rfloor+1}|z_{j - \lfloor q/2 \rfloor})}{p(z_{j - \lfloor q/2 \rfloor+1}|z_{j - \lfloor q/2 \rfloor})}
+\cdots+ \frac{p'(z_{j}|z_{j- \lfloor q/2 \rfloor}^{j-1})}{p(z_{j}|z_{j- \lfloor q/2 \rfloor}^{j-1})} \right) \log p(z_{j}|z_{j-\lfloor q/2 \rfloor}^{j-1}) \\
& =-\sum_{z_{j-\lfloor q/2 \rfloor}^j} p'(z_{j-\lfloor q/2 \rfloor}^j) \log p(z_{j}|z_{j-\lfloor q/2 \rfloor}^{j-1}).
\end{align*}
Then, following~\cite{gm05}, we can prove that for any small $\eps$, we have
$$
\sum_{z_1^n} |p'(z_n|z_1^{n-1})| = O((1+\eps)^n).
$$
This, together with Condition (\ref{simulator}.c), implies that for some $0 < \rho_1 < 1$,
$$
E[W_j]-H'(Z)=O(\rho_1^{\lfloor q/2 \rfloor}),
$$
which further implies that for some $0 < \rho_1 < 1$
$$
\frac{S_n}{k p}-H'(Z)=\frac{E[S_n]-k p H'(Z)}{k p} = \frac{\sum_j(W_j-H'(Z))}{k p} = O(\rho_1^{\lfloor q/2 \rfloor}).
$$
\end{proof}

The following large deviation type lemma gives a sub-exponentially decaying upper bound on the tail probability of $g_n(X_1^n, Y_1^n)$ deviating from $I'(X; Y)$.
\begin{thm} \label{Chernoff-Bound}
For any $\eps > 0$, there exists some $0 < \gamma, \delta < 1$ such that ,
$$
P\left( \left| g_n(X_1^n, Y_1^n) - I'(X; Y) \right| \geq \eps \right) \leq \gamma^{n^{\delta}}.
$$
\end{thm}

\begin{proof}
By Lemma~\ref{expectation-convergence} and Remark~\ref{Z-Created}, it suffices to prove that for any $Z$ satisfying Conditions (\ref{simulator}.a)-(\ref{simulator}.c) and for any $\eps > 0$, there exist $0 < \gamma, \delta < 1$ such that
\begin{equation}  \label{absolute-value}
P\left( \left| \frac{S_n-E[S_n]}{k p} \right| \geq \eps \right) \leq \gamma^{n^{\delta}}.
\end{equation}
By the Markov inequality, we have
\begin{equation} \label{MI-Snn}
P(S_n-E[S_n] \geq k p \eps) = P\left(\frac{t (S_n-E[S_n])}{p^2} \geq \frac{t k p \eps}{p^2}\right) \leq \frac{E[e^{t (S_n-E[S_n])/p^2]}}{e^{t k \eps/p}}.
\end{equation}
As in~\cite{han11}, applying Conditions (\ref{simulator}.a) and (\ref{simulator}.b), we then have
\begin{align} \label{first-iteration}
\nonumber E[e^{t (S_n-E[S_n])/p^2}] & =E[e^{t \sum_{i=1}^{k-1} (\zeta_i-E[\zeta_i])/p^2} e^{t (\zeta_k-E[\zeta_k])/p^2}] \\
& = (1 + O(\lambda^{q(n)/2})) E[e^{t \sum_{i=1}^{k-1} (\zeta_i-E[\zeta_i])/p^2}] E[e^{t \zeta_k}],
\end{align}
for some $0 < \lambda < 1$. An iterative application of (\ref{first-iteration}) yields that for any $0 < t < 1$
\begin{eqnarray} \label{Sn-prime}
\nonumber E[e^{t (S_n-E[S_n])/p^2}]&=&E[e^{t \sum_{i=1}^k (\zeta_i-E[\zeta_i])/p^2}] \\
&=& (1 + O(\lambda^{q(n)/2}))^{k-1} (E[e^{t (\zeta_1-E[\zeta_1])/p^2}])^k,
\end{eqnarray}
as $n$ goes to infinity. By Condition (\ref{simulator}.a), we have
$$
\zeta_1-E[\zeta_1]=O(p q), \mbox{ and thus, } O((\zeta_1-E[\zeta_1])^2/p^4)= O(q^2/p^2)=o(1).
$$
It then follows that for any $0 < t < 1$,
$$
E[e^{t (\zeta_1-E[\zeta_1])/p^2}]=1+ o(1)t^2.
$$
Choosing $t=n^{-(1-\alpha)/2}$, then, by (\ref{MI-Snn}) and (\ref{Sn-prime}), we deduce that
\begin{eqnarray}
\nonumber P\left( \frac{S_n-E[S_n]}{k p}  \geq \eps \right) & \leq & \frac{E[e^{t (S_n-E[S_n])/p^2]}}{e^{t k \eps/p}} \\
\nonumber &\leq &  (1+O(\lambda^{q(n)/2}))^k \frac{(1+o(1)t^2)^{n^{1-\alpha}}}{(1+t \eps +O(1)t^2)^{n^{1-2 \alpha}}}\\
\nonumber & = & O(e^{-n^{1/2-3\alpha/2}}).
\end{eqnarray}
Noticing that $0 < \alpha < 1/3$ (and thus $1/2-3\alpha/2 < 0$), we conclude that for any $\eps > 0$, there exists $0 < \gamma, \delta < 1$ such that
$$
P\left(\frac{S_n-E[S_n]}{k p} \geq \eps \right) \leq \gamma^{n^{\delta}}.
$$
With a parallel argument, one verifies that for any $\eps > 0$, there exists $0 < \gamma, \delta < 1$ such that
$$
P\left(\frac{S_n-E[S_n]}{k p} \leq -\eps \right) \leq \gamma^{n^{\delta}},
$$
which immediately implies (\ref{absolute-value}). The proof is then complete.
\end{proof}

The following theorem states that our simulator is asymptotically unbiased.
\begin{thm} With probability $1$,
$$
g_n(X_1^n, Y_1^n) \to I'(X; Y),
$$
as $n$ tends to $\infty$.
\end{thm}

\begin{proof}
It immediately follows from Theorem~\ref{Chernoff-Bound} and the Borel-Cantelli lemma.
\end{proof}

\begin{rem}
In our notation, the following expression has been proposed in~\cite{pa04} as a simulator of $I'(X; Y)$:
$$
H(X_2|X_1) -\frac{p'(Y_{1}^n)}{p(Y_{1}^n)} \log p(Y_{1}^n)/n + \frac{p'(X_1^n, Y_{1}^n)}{p(X_1^n, Y_{1}^n)} \log p(X_1^n, Y_{1}^n)/n.
$$
Extensive numerical experiments conducted in~\cite{pa04} suggest that this simulator converges to $I'(X; Y)$ almost surely as $n$ tends to infinity, however, there is no rigorous proof for the convergence.
\end{rem}

\section{Convergence}  \label{f-convergence}

Consider the following condition: 
\begin{enumerate} [label=(\thesection.\alph*)]
\item $P(\cap_{k=1}^{\infty} \cup_{n=k}^{\infty} A_n)=0$, that is, $A_n, n \in \mathbb{N}$, only occurs finitely many times,
\end{enumerate}
which will be assumed throughout the convergence analysis in the paper. Particularly, in this section, assuming (\thesection.a), we will show that $\{I(X(\theta_n); Y(\theta_n))\}$ converges almost surely. Note that if $\theta=\mathbb{R}^d$, then Assumption (\thesection.a) will be trivially satisfied and the iteration in (\ref{theta-iteration}) can be simply written as
\begin{equation} \label{simplified-theta-iteration}
\theta_{n+1}=\theta_n + a_n g_{n^b}(\theta_n).
\end{equation}
In fact, unless specified otherwise, we will simply assume that $\theta=\mathbb{R}$ in all the proofs in this paper to avoid obscuring the main idea. The proofs of the same results under Assumption (\thesection.a) follow from parallel arguments only with an increasing level of notational complexity.

Henceforth, we will write
$$
f(\theta)=I(X(\theta); Y(\theta)), \quad f_n(\theta)=I_n(X(\theta); Y(\theta)).
$$
Note that under Assumption (\ref{channel-model}.a), Theorem $1.1$ of~\cite{gm05} implies that
\begin{itemize}
\item[] $f(\theta)$ is analytic and each of its derivatives is uniformly  bounded over all $\theta \in \Theta$,
\end{itemize}
a key fact that will be implicitly used throughout the paper. Now, rewrite (\ref{simplified-theta-iteration}) as
\begin{equation} \label{R-n-introduced}
\theta_{n+1}=\theta_n + a_n f'(\theta_n)+ a_n R_n(\theta_n),
\end{equation}
where
$$
R_n(\theta_n) \triangleq g_{n^b}(\theta_n)-f'(\theta_n).
$$
It can be easily verified that
\begin{align}  \label{hat-R-n-introduced}
\nonumber f(\theta_{n+1})-f(\theta_n) &= \int_0^1 f'(\theta_{n}+t(\theta_{n+1}-\theta_n)) (\theta_{n+1}-\theta_n) dt \\
\nonumber & =\int_0^1 f'(\theta_n) (\theta_{n+1}-\theta_n) dt + \int_0^1 (f'(\theta_{n}+t(\theta_{n+1}-\theta_n))-f'(\theta_n)) (\theta_{n+1}-\theta_n) dt \\
\nonumber & =a_n f'(\theta_n) (f'(\theta_n)+R_n (\theta_n))+ \int_0^1 (f'(\theta_{n}+t(\theta_{n+1}-\theta_n))-f'(\theta_n)) (\theta_{n+1}-\theta_n) dt \\
& =a_n f'^2(\theta_n) + \hat{R}_n(\theta_n),
\end{align}
where
$$
\hat{R}_n(\theta_n) \triangleq a_n f'(\theta_n) R_n (\theta_n)+ \int_0^1 (f'(\theta_{n}+t(\theta_{n+1}-\theta_n))-f'(\theta_n)) (\theta_{n+1}-\theta_n) dt.
$$

\begin{lem} \label{convergent-even-on-some-event}
$\sum_{n=0}^{\infty} \hat{R}_n(\theta_n)$ converges almost surely.
\end{lem}

\begin{proof}
Let
$$
T_1=\sum_{n=0}^{\infty} a_n f'(\theta_n) R_n (\theta_n), ~~ T_2=\sum_{n=0}^{\infty} \int_0^1 (f'(\theta_{n}+t(\theta_{n+1}-\theta_n))-f'(\theta_n)) (\theta_{n+1}-\theta_n) dt.
$$
It suffices to prove that $T_1, T_2$ both converge almost surely.

For $T_1$, note that
\begin{align*}
T_1 & =\sum_{n=0}^{\infty} a_n f'(\theta_n) (g_{n^{b}}(\theta_n)-f'(\theta_n)) \\
& =\sum_{n=0}^{\infty} a_n f'(\theta_n) (g_{n^{b}}(\theta_n)-f'_{n^{b}}(\theta_n))+\sum_{n=0}^{\infty} a_n f'(\theta_n)(f'_{n^{b}}(\theta_n)-f'(\theta_n)).
\end{align*}
It follows from Theorem~\ref{expectation-convergence} that there exists $0 < \rho_0 < 1$ such that
\begin{equation} \label{close-1}
\sum_{n=0}^{\infty} a_n |f'(\theta_n)| |(f'_{n^{b}}(\theta_n)-f'(\theta_n))| \leq \sum_{n=0}^{\infty} a_n |f'(\theta_n)| \rho_0^{n^{b}} < \infty.
\end{equation}
Then, using Lemma~\ref{order-of-mements}, one verifies that uniformly over all $\theta_n \in \Theta$,
\begin{equation} \label{conditional-square}
\sum_{n=0}^{\infty} E[\{a_n^2 (f'(\theta_n))^2 R_n^2 (\theta_n)\}] = \sum_{n=0}^{\infty} O\left(\frac{1}{n^{2a+b(1-3\beta)}}\right),
\end{equation}
which converges since $2a+b-3 b \beta > 1$. Noting that $\{a_n f'(\theta_n) R_n (\theta_n), \mathcal{B}(X_1^n)\}$ is a Martingale difference sequence and applying Doob's Martingale convergence theorem (see Theorem $2.8.7$ of~\cite{st74}), we deduce that
$$
\sum_{n=0}^{\infty} a_n f'(\theta_n) (g_{n^{b}}(\theta_n)-f'_{n^{b}}(\theta_n))
$$
converges with probability $1$. The almost sure convergence of $T_1$ then follows.

For $T_2$, it is easy to check that
$$
\left| \int_0^1 (f'(\theta_{n}+t(\theta_{n+1}-\theta_n))-f'(\theta_n)) (\theta_{n+1}-\theta_n) dt \right| = O((\theta_{n+1}-\theta_{n})^2) = O(a_n^2 (f'(\theta_n))^2)+ O(a_n^2 R_n^2(\theta_n)).
$$
Similarly as in deriving (\ref{close-1}) and (\ref{conditional-square}), we have
$$
\sum_{n=0}^{\infty} a_n^2 (f'_{n^{b}}(\theta_n)-f'(\theta_n))^2 < \infty, ~~ \sum_{n=0}^{\infty} E[a_n^2 (g_{n^{b}}(\theta_n)-f'_{n^{b}}(\theta_n))^2] < \infty,
$$
and furthermore,
$$
\sum_{n=0}^{\infty} a_n^2 (g_{n^{b}}(\theta_n)-f'_{n^{b}}(\theta_n))^2
$$
converges almost surely. This, together with (\ref{close-1}), further implies that
$$
\sum_{n=0}^{\infty} a_n^2 |(g_{n^{b}}(\theta_n)-f'_{n^{b}}(\theta_n))(f'_{n^{b}}(\theta_n)-f'(\theta_n)|
$$
converges almost surely. Recalling that
$$
R_n(\theta_n)=g_{n^b}(\theta_n)-f'_{n^b}(\theta_n)+ f'_{n^b}(\theta_n)-f'(\theta_n),
$$
we conclude that
$$
\sum_{n=0}^{\infty} a_n^2 R_n^2(\theta_n) < \infty,
$$
which further implies that
$$
\sum_{n=0}^{\infty} \int_0^1 (f'(\theta_{n}+t(\theta_{n+1}-\theta_n))-f'(\theta_n)) (\theta_{n+1}-\theta_n) dt
$$
converges almost surely. The proof is then complete.
\end{proof}

We are now ready for the following convergence theorem, whose proof closely follows that of Lemma $7$ in~\cite{ta10}, which can be further traced back to the standard proof of the Martingale convergence theorem~\cite{st74}.

\begin{thm} \label{convergence-theorem}
With probability $1$, we have
$$
\lim_{n \to \infty} f'(\theta_n)=0 \mbox{ and } \lim_{n \to \infty} f(\theta_n) \mbox{ exists }.
$$
\end{thm}

\begin{proof}

Recall that
$$
f(\theta_{n+1})-f(\theta_n)=a_n f'^2(\theta_n) + \hat{R}_n (\theta_n),
$$
an iterative application of which implies
$$
f(\theta_n)=f(\theta_0)+\sum_{i=0}^{n-1} a_i (f'(\theta_i))^2+\sum_{i=0}^{n-1} \hat{R}_i(\theta_i).
$$
Applying Lemma~\ref{convergent-even-on-some-event}, we deduce that with probability $1$,
$$
\sum_{i=0}^{\infty} a_i (f'(\theta_i))^2 < \infty,
$$
which, in return, implies that $\lim_{n \to \infty} f(\theta_n)$ exists and furthermore there is a subsequence $\{\theta_{n_j}\}$ such that $f'(\theta_{n_j})$ converges to $0$ as $j$ tends to infinity.

We now prove that
$$
\lim_{n \to \infty} f'(\theta_n) = 0.
$$
By way of contradiction, suppose otherwise. Then, there exists $\eps > 0$ such that there exist infinite sequences $m_k, n_k, k=1, 2, \cdots$, such that
\begin{equation} \label{crossing-conditions}
|f'(\theta_{m_k})| \leq \eps, ~~ |f'(\theta_{n_k})| \geq 2 \eps, ~~ |f'(\theta_{i})| \geq \eps
\end{equation}
for all $m_k+1 \leq i \leq n_k$. It then follows that
\begin{align} \label{mknk}
\nonumber \eps & \leq |f'(\theta_{n_k}) - f'(\theta_{m_k})| \\
\nonumber & = O(|\theta_{n_k}-\theta_{m_k}|) \\
\nonumber & = O\left(\sum_{i=m_k}^{n_k-1} a_i |f'(\theta_i)|\right) + O\left(\left|\sum_{i=m_k}^{n_k-1} a_i R_i(\theta_i) \right| \right) \\
& = O\left(\sum_{i=m_k}^{n_k-1} a_i \right)+O \left( \left|\sum_{i=m_k}^{n_k-1} a_i R_i(\theta_i) \right| \right).
\end{align}
As in the proof of Lemma~\ref{convergent-even-on-some-event}, using the decomposition
$$
R_n(\theta_n)=g_{n^{b}}(\theta_n)-f'(\theta_n)=g_{n^{b}}(\theta_n)-f'_{n^{b}}(\theta_n)+f'_{n^{b}}(\theta_n)-f'(\theta_n),
$$
we deduce that $\sum_{n=0}^{\infty} a_n R_n(\theta_n)$ converges almost surely, and hence $\left| \sum_{i=m_k}^{n_k-1} a_i R_i(\theta_i) \right|$
tends to $0$ as $k$ goes to $\infty$. On the other hand, by (\ref{crossing-conditions}), we have
$$
\eps^2 \sum_{i=m_k}^{n_k-1} a_i \leq \sum_{i=m_k}^{\infty} a_i (f'(\theta_i))^2.
$$
This implies that as $k$ tends to $\infty$, $\sum_{i=m_k}^{n_k-1} a_i$
tends to zero, which, together with (\ref{mknk}), further implies that
$$
\eps \leq \lim_{k \to \infty} |f'(\theta_{n_k})-f'(\theta_{m_k})|=0,
$$
a contradiction.
\end{proof}

\begin{rem}
The fact that $\{f(\theta_n)\}$ converges almost surely does not necessarily imply that $\{\theta_n\}$ converges almost surely. In the remainder of this paper, we will prove, under some assumptions, that $\{\theta_n\}$ does converge almost surely.
\end{rem}

\section{Some Estimations}  \label{some-estimations}

In this section, assuming (\ref{f-convergence}.a), we will derive some estimations that will be used in the later sections for convergence analysis.

For any $j \in \mathbb{N}$, let
$$
A_j=a_1+a_2+\cdots+a_{j-1},
$$
and for any $h > 0$ and any $n \in \mathbb{N}$, define
$$
t(n, h) \triangleq \min\{k: a_{n}+a_{n+1}+\cdots+a_{k-1} \geq h\}.
$$
Now, for any fixed $n_0 \in \mathbb{N}$, recursively define
$$
n_{k+1}=t(n_k, h).
$$
One then verifies that for $k$ sufficiently large,
\begin{equation}  \label{recursive-n-k}
A_{n_{k+1}}-A_{n_k}=\hat{O}(h), ~~ n_k=\hat{O}(k^{1/(1-a)}),
\end{equation}
where by $M = \hat{O}(N)$, we mean that there exist positive constants $C_1, C_2$ such that
$$
C_1 N \leq M \leq C_2 N.
$$

Now, an iterated application of
$$
\theta_{n+1}-\theta_n=a_n f'(\theta_n)+a_n R_n(\theta_n)
$$
yields
\begin{align*}
\theta_k & =\theta_n+\sum_{i=n}^{k-1} a_i f'(\theta_i)+\sum_{i=n}^{k-1} a_i R_i (\theta_i) \\
& =\theta_n+(A_k-A_n) f'(\theta_n)+\sum_{i=n}^{k-1} a_i R_i (\theta_i)+\sum_{i=n}^{k-1} a_i (f'(\theta_i)-f'(\theta_n)) \\
& =\theta_n+R_{n, k},
\end{align*}
where
\begin{equation}  \label{R-n-k}
R_{n, k}=\sum_{i=n}^{k-1} a_i R_i (\theta_i)+\sum_{i=n}^{k-1} a_i (f'(\theta_i)-f'(\theta_n)).
\end{equation}
Similarly, an iterated application of
$$
f(\theta_{n+1})-f(\theta_n)=a_n f'^2(\theta_n) + \hat{R}_n (\theta_n)
$$
yields
\begin{align} \label{f-k-n}
\nonumber f(\theta_{k})-f(\theta_n) & =\int_0^1 f'(\theta_{n}+t(\theta_{k}-\theta_n)) (\theta_{k}-\theta_n) dt\\
\nonumber & =\int_0^1 f'(\theta_n) (\theta_{k}-\theta_n) dt + \int_0^1 (f'(\theta_{n}+t(\theta_{k}-\theta_n))-f'(\theta_n)) (\theta_{k}-\theta_n) dt \\
\nonumber & =f'(\theta_n) ((A_k-A_n) f'(\theta_n)+R_{n, k})+ \int_0^1 (f'(\theta_{n}+t(\theta_{k}-\theta_n))-f'(\theta_n)) (\theta_{k}-\theta_n) dt \\
\nonumber & =(A_k-A_n) f'^2(\theta_n) + f'(\theta_n) R_{n, k} + \int_0^1 (f'(\theta_{n}+t(A_k-A_n))-f'(\theta_n)) (A_k-A_n) dt \\
& = (A_k - A_n) f'^2(\theta_n) + \hat{R}_{n, k} (\theta_n),
\end{align}
where
\begin{equation} \label{hat-R-n-k}
\hat{R}_{n, k} (\theta_n)=f'(\theta_n) R_{n, k} + \int_0^1 (f'(\theta_{n}+t(A_k-A_n))-f'(\theta_n)) (A_k-A_n) dt.
\end{equation}

The following lemma introduces a positive random variable, $\tilde{C}_0$, and a constant, $\tau$, which will be referred to throughout the rest of the paper.

\begin{lem} \label{random-O-term}
There exists a positive random variable $\tilde{C}_0$ such that for all $n$ and for any $\tau > 0$ with $2a+b-3 b \beta - 2\tau > 1$,
$$
\sup_{k \geq n} \left| \sum_{i=n}^k a_i R_i(\theta_i) \right| \leq \tilde{C}_0 n^{-\tau} \mbox{ a.s. }
$$
\end{lem}

\begin{proof}

For any $\tau > 0$ with $2a+b-3 b \beta - 2 \tau > 1$, as in the proof of Lemma~\ref{convergent-even-on-some-event}, we deduce that $\sum_{i=1}^{\infty} i^{\tau} a_i R_i$ converges almost surely. Letting
$$
T_n \triangleq \sum_{i=1}^{n} i^{\tau} a_i R_i(\theta_i),
$$
we then have for any $k \geq n$,
\begin{align*}
\sum_{i=n}^k a_i R_i(\theta_i)
& = \sum_{i=n}^k (i^{\tau} a_i R_i(\theta_i)) i^{-\tau}
\\
& =\sum_{i=n}^k (T_i - T_{i-1}) i^{-\tau} \\
&=\sum_{i=n}^k T_i i^{-\tau}-\sum_{i=n+1}^k T_{i-1} i^{-\tau}\\
&=\sum_{i=n}^k T_i i^{-\tau}-\sum_{i=n}^{k-1} T_{i} (i+1)^{-\tau}\\
&=T_k k^{-\tau}+ \sum_{i=n}^{k-1} (i^{-\tau}-(i+1)^{-\tau}) T_i\\
& \leq (k^{-\tau}+\sum_{i=n}^{k-1} (i^{-\tau}-(i+1)^{-\tau})) \sup_i T_i \\
& = n^{-\tau} \sup_i T_i,
\end{align*}
which immediately implies the lemma.
\end{proof}

In the following, to avoid notational cumbersomeness, we will use $C$ to denote a positive constant, which may not be the same on its each appearance.
\begin{lem} \label{EstimationsOfDifferences}
Let $0 < h < 1$ and $\tilde{C}_0, \tau$ be as in Lemma~\ref{random-O-term}, then we have
\begin{enumerate} [label=(\arabic*)]
\item there exists a constant $C > 0$ such that
$$
|f'(\theta_{t(n, h)})| \leq C (\tilde{C}_0 n^{-\tau}+ |f'(\theta_n)|).
$$
\item there exists a constant $C > 0$ such that
$$
|\theta_{t(n, h)}- \theta_n| \leq C (\tilde{C}_0 n^{-\tau} + h |f'(\theta_n)|).
$$
\item there exists a constant $C > 0$ such that
$$
|R_{n, t(n, h)}| \leq C(\tilde{C}_0 n^{-\tau} +  h^2 |f'(\theta_n)|).
$$
\item there exists a constant $C > 0$ such that
$$
|\hat{R}_{n, t(n, h)}| \leq C (\tilde{C}_0^2 n^{-2\tau} + \tilde{C}_0 n^{-\tau} |f'(\theta_n)|+h^2 |f'(\theta_n)|^2).
$$
\item there exists a constant $C > 0$ such that
$$
f(\theta_n)-f(\theta_{t(n, h)}) \leq -(3/4-3 C h/2)h |f'(\theta_n)|^2+ C \tilde{C}_0^2 n^{-2 \tau} (1+1/(2 h^2)).
$$

\item there exists $C > 0$ such that for sufficiently small $h$
$$
2(f(\theta_n)-f(\theta_{t(n, h)})) + |f'(\theta_n)| |\theta_{t(n, h)}-\theta_n| \leq (C+1/(2h^2)) \tilde{C}_0^2 n^{-2 \tau}.
$$

\item for any $\tau' < \tau$, there exists a positive constant $C$ such that for sufficiently small $h$, we have
$$
|\theta_{t(n, h)}-\theta_n| \leq C n^{\tau'} (f(\theta_{t(n, h)})-f(\theta_n))+C \tilde{C}_0^2 n^{-\tau'}.
$$
\end{enumerate}
\end{lem}

\begin{proof}

In this proof, for notational simplicity, we will write $t(n, h)$ as $k$.

Note that there exists a positive constant $C$ such that
\begin{align*}
|f'(\theta_k)| & \leq |f'(\theta_n)|+|f'(\theta_k)-f'(\theta_n)| \\
& \leq |f'(\theta_n)|+C|\theta_k-\theta_n| \\
& \leq |f'(\theta_n)|+ C \sum_{i=n}^{k-1} a_i |f'(\theta_i)| + C |\sum_{i=n}^{k-1} a_i R_i(\theta_i)|,
\end{align*}
where we have applied (\ref{R-n-introduced}).
Applying Lemma~\ref{random-O-term}, we then have
$$
|f'(\theta_k)| \leq C \tilde{C}_0 n^{-\tau} + |f'(\theta_n)|+C \sum_{i=n}^{k-1} a_i |f'(\theta_i)|.
$$
Applying Gronwall's lemma, we then have for $n$ sufficiently large
$$
|f'(\theta_k)| \leq (C \tilde{C}_0 n^{-\tau}+|f'(\theta_n)|) \exp(C(a_n+a_{n+1}+\cdots+a_{k-1})) \leq \exp(C) (C \tilde{C}_0 n^{-\tau}+ |f'(\theta_n)|),
$$
where we have used the fact that for $n$ large enough
$$
a_n+a_{n+1}+\cdots+a_{k-1} \approx h < 1.
$$
We have then established (1).

It then follows from (1) that for some $C$
\begin{align*}
|\theta_k - \theta_n| & \leq \sum_{i=n}^{k-1} a_i |f'(\theta_i)|+ |\sum_{i=n}^{k-1} a_i R_i(\theta_i)| \\
& \leq (A_k - A_n) (C \tilde{C}_0 n^{-\tau} + C |f'(\theta_n)|) + \tilde{C}_0 n^{-\tau},
\end{align*}
which immediately implies (2).

Now, by (\ref{R-n-k}) and (2), we have for some $C$
\begin{align*}
|R_{n, k}| & \leq \tilde{C}_0 n^{-\tau}+ C\sum_{i=n}^{k-1} a_i |\theta_i - \theta_n| \\
& \leq \tilde{C}_0 n^{-\tau} + C^2 (A_k - A_n) (\tilde{C}_0 n^{-\tau} + (A_k - A_n) |f'(\theta_n)|),
\end{align*}
which establishes (3).

Furthermore, by (\ref{hat-R-n-k}), (2) and (3), we have
\begin{align*}
|\hat{R}_{n, k}| & \leq |f'(\theta_n)| |R_{n, k}|+ C |\theta_k-\theta_n|^2 \\
& \leq C \tilde{C}_0 n^{-\tau} |f'(\theta_n)| + C (A_k-A_n)^2 |f'(\theta_n)|^2 + 2 C^3 (\tilde{C}_0^2 n^{-2\tau} + (A_k - A_n)^2 |f'(\theta_n)|^2),
\end{align*}
which establishes (4).

It then follows from (\ref{f-k-n}), (III) and (IV) and  that for sufficiently large $n$
\begin{align*}
f(\theta_n)-f(\theta_{k}) & \leq -(A_k-A_n) |f'(\theta_n)|^2+|\hat{R}_{n, k}| \\
& \leq -3h/4  |f'(\theta_n)|^2+ C (\tilde{C}_0^2 n^{-2\tau}+ \tilde{C}_0 n^{-\tau} |f'(\theta_n)|+h^2|f'(\theta_n)|^2) \\
& \leq -3h/4  |f'(\theta_n)|^2+ C (\tilde{C}_0^2 n^{-2\tau}+ \tilde{C}_0^2 n^{-2\tau}/(2 h^2) + h^2 |f'(\theta_n)|^2/2 + h^2 |f'(\theta_n)|^2) \\
& \leq  -3h/4  |f'(\theta_n)|^2+C (\tilde{C}_0^2 n^{-2 \tau} (1+1/(2 h^2)) + 3 h^2/2 |f'(\theta_n)|^2 )\\
& \leq -(3/4-3Ch/2) h |f'(\theta_n)|^2+ C \tilde{C}_0^2 n^{-2 c} (1+1/(2 h^2)),
\end{align*}
which establishes (5).

It follows from (\ref{f-k-n}), (3) and (4) that
\begin{align*}
f(\theta_{k})-f(\theta_n) & =|f'(\theta_n)| |(A_{k}-A_n) f'(\theta_n)|+\hat{R}_{n, k}\\
& =|f'(\theta_n)| |\theta_{k}-\theta_n+R_{n, k}|+\hat{R}_{n, k} \\
&\geq |f'(\theta_n)| (|\theta_{k}-\theta_n|-|R_{n, k}|) -\hat{R}_{n, k}\\
&\geq |f'(\theta_n)| |\theta_{k}-\theta_n| - C (\tilde{C}_0^2 n^{-2\tau} + \tilde{C}_0 n^{-\tau} |f'(\theta_n)|+(A_{k} - A_n)^2 |f'(\theta_n)|^2),
\end{align*}
which implies that
\begin{align*}
f(\theta_n)-f(\theta_{k}) + |f'(\theta_n)| |\theta_{k}-\theta_n| & \leq
C (\tilde{C}_0^2 n^{-2\tau} + \tilde{C}_0 n^{-\tau} |f'(\theta_n)|+(A_{k} - A_n)^2 |f'(\theta_n)|^2)\\
& \leq C (\tilde{C}_0^2 n^{-2 \tau} (1+1/(2 h^2)) + 3 h^2/2 |f'(\theta_n)|^2 ).
\end{align*}
Applying (V), we then have for sufficiently small $h$,
$$
f(\theta_n)-f(\theta_{k}) + |f'(\theta_n)| |\theta_{k}-\theta_n|
\leq 2 C (1+1/(2h^2)) \tilde{C}_0^2 n^{-2 \tau} +f(\theta_{k}-f(\theta_n)),
$$
which can be rewritten as
$$
2(f(\theta_n)-f(\theta_{k})) + |f'(\theta_n)| |\theta_{k}-\theta_n| \leq 2 C (1+1/(2h^2)) \tilde{C}_0^2 n^{-2 \tau},
$$
which establishes (6).

We next prove (7). If $|f'(\theta_n)| \leq n^{-\tau'}$, applying (II), we deduce that
\begin{equation} \label{smaller-than}
|\theta_{k}-\theta_n| \leq C \tilde{C}_0 n^{-\tau'} + C h^2 n^{-\tau'}.
\end{equation}
It follows from (\ref{f-k-n}) and (4) that
\begin{align*}
|f(\theta_k)-f(\theta_n)| & \leq (A_k - A_n) f'^2(\theta_n) + |\hat{R}_{n, k} (\theta_n)| \\
& \leq (A_k - A_n) f'^2(\theta_n)+ C (\tilde{C}_0^2 n^{-2\tau} + \tilde{C}_0 n^{-\tau} |f'(\theta_n)|+h^2 |f'(\theta_n)|^2)\\
& \leq (A_k - A_n) f'^2(\theta_n)+ C (\tilde{C}_0^2 n^{-2\tau'} + \tilde{C}_0 n^{-\tau'} |f'(\theta_n)|+h^2 |f'(\theta_n)|^2)\\
& \leq C (\tilde{C}_0^2 n^{-2 \tau'} (1+1/(2 h^2))) + (h+3 C h^2/2) |f'(\theta_n)|^2,
\end{align*}
which, together with (\ref{smaller-than}), immediately implies that for some $C$,
\begin{align} \label{dico-1}
\nonumber |\theta_{k}-\theta_n| & \leq n^{\tau'} (f(\theta_k)-f(\theta_n)) + n^{\tau'} |f(\theta_k)-f(\theta_n)| + C \tilde{C}_0 n^{-\tau'} + C h^2 n^{-\tau'} \\
\nonumber & \leq n^{\tau'} (f(\theta_k)-f(\theta_n)) + C (\tilde{C}_0^2 n^{-2 \tau'} (1+1/(2 h^2)))\\
& + (h+3 C h^2/2) |f'(\theta_n)|^2 + C \tilde{C}_0 n^{-\tau'} + C h^2 n^{-\tau'}.
\end{align}
On the other hand, if $|f'(\theta_n)| \geq n^{-\tau'}$, applying (6), we deduce that
\begin{align} \label{dico-2}
\nonumber |\theta_{k}-\theta_n| & \leq 2 |f'(\theta_n)|^{-1} (f(\theta_{k}-f(\theta_n))) + (C+1/(2h^2)) |f'(\theta_n)|^{-1}  \tilde{C}_0^2 n^{-2 \tau} \\
& \leq 2 n^{\tau'} (f(\theta_{k}-f(\theta_n))) + (C+1/(2h^2)) \tilde{C}_0^2 n^{-\tau'}.
\end{align}
Combining (\ref{dico-1}) and (\ref{dico-2}), we then have established (7).
\end{proof}

\section{Rate of Convergence with Concavity} \label{with-concavity}

In this section, we assume that
\begin{enumerate} [label=(\thesection.\alph*)]
\item $f(\theta)$ is strictly concave with respect to $\theta$. More precisely, there exists $\hat{\epsilon} > 0$ such that for any $\theta_1, \theta_2 \in \Theta$,
$$
f'_t(t \theta_1 + (1-t) \theta_2) \geq \hat{\epsilon},
$$
for all $0 \leq t \leq 1$.
\item With probability $1$, $\theta_n$ converges to the unique global maximum $\theta^*$ as $n$ tends to $\infty$.
\end{enumerate}
Here, let us note that (\thesection.a), together with Theorem~\ref{Chernoff-Bound}, implies (\ref{f-convergence}.a). With Assumptions (\thesection.a) and (\thesection.b), which, as argued in Section~\ref{Memoryless-Channels}, can be satisfied for a class of finite-state channels, we will derive the convergence rate of $\{\theta_n\}$. Again, for notational convenience only, we assume that $\Theta=\mathbb{R}$ in the proofs.

From
$$
\theta_{n+1}-\theta_n=a_n f'(\theta_n)+a_n R_n(\theta_n),
$$
trivially we have
$$
\Delta_{n+1}-\Delta_n = -a_n f'(\theta_n)-a_n R_n(\theta_n),
$$
where
$$
\Delta_n \triangleq (\theta^*-\theta_{n}).
$$
It immediately from the above two conditions that for $\theta$ sufficiently close to $\theta^*$
\begin{equation} \label{Lojasiewicz-1}
f(\theta)=\hat{O} (|\theta^*-\theta|^2), \qquad f'(\theta)=\hat{O} (|\theta^*-\theta|).
\end{equation}
So, if $\theta_n$ is sufficiently close to $\theta^*$, we will have
$$
f(\theta_n)=\hat{O} (\Delta_n^2), \qquad f'(\theta_n)=\hat{O} (|\Delta_n|).
$$

Throughout the paper, by $M=\tilde{O}(N)$, we mean that there exists a positive random variable $\tilde{C}$ such that with probability $1$,
$$
|M| \leq \tilde{C} N.
$$
In this section, we will prove that $\Delta_n$ is at most of order $\tilde{O}(n^{-\tau})$.

We first prove the following lemma.

\begin{lem} \label{lim-inf}
There exists $l \in \mathbb{N}$ such that
$$
\liminf_{n \to \infty} n^{\tau} |\Delta_n| \leq l \tilde{C}_0.
$$
\end{lem}

\begin{proof}
Suppose, by way of contradiction, that for any $l$,
\begin{equation} \label{not-too-small-2}
n^{\tau} |\Delta_n| \geq l \tilde{C}_0,
\end{equation}
as long as $n$ is sufficiently large. First, pick $n_0$ sufficiently large such that (\ref{not-too-small-2}) is satisfied and then recursively define
$$
n_{k+1}=t(n_k, h).
$$
for some $0 < h < 1$. We then have, for any feasible $k$,
$$
\theta_{n_{k+1}}=\theta_{n_k}+ (A_{n_{k+1}}-A_{n_k}) f'(\theta_{n_k})+ R_{n_k, n_{k+1}}.
$$
It then follows from Lemma~\ref{EstimationsOfDifferences} (3) and (\ref{Lojasiewicz-1}) that
$R_{n_k, n_{k+1}}$ is dominated by $|f'(\theta_{n_k})|$ as long as $l$ is chosen sufficiently large and $h$ is chosen sufficiently small. Noticing that due to the concavity of $f$, $\Delta_n$ always has the same sign as $f'(\theta_n)$, then we have
$$
|\Delta_{n_{k+1}}| \leq  |\Delta_{n_k}|- h/2 |\Delta_{n_k}| \leq |\Delta_{n_k}| e^{-h/2},
$$
an iterative application of which would yield
$$
\Delta_{n_k} \leq \Delta_{n_0} e^{-k h/2}.
$$
It then follows that for any $k$
$$
\Delta_{n_0} n_k^{\tau} e^{-k h/2} \geq n_k^{\tau} \Delta_{n_k} \geq l \tilde{C}_0.
$$
This, together with the fact that (see (\ref{recursive-n-k}))
$$
n_k= \hat{O}(k^{1/(1-a)}),
$$
as $k$ tends to infinity, implies that
$$
\tilde{C}_0 \leq 0,
$$
which is a contradiction.
\end{proof}

\begin{thm} \label{lim-sup}
$$
|\Delta_n| = \tilde{O}(n^{-\tau}).
$$
\end{thm}

\begin{proof}

It is enough to prove that there exists an integer $l$ such that for all $n$ sufficiently large,
$$
n^{\tau} |\Delta_n| \leq l \tilde{C}_0.
$$
By way of contradiction, suppose otherwise. Then, by Lemma~\ref{lim-inf}, for any $l$ and arbitrarily large $N$, we can find $k_0 > m_0 > N$ such that
$$
m_0^{\tau}  \Delta_{m_0} \leq 2 l \tilde{C}_0, \quad k_0^{\tau} \Delta_{k_0}  \geq 3 l \tilde{C}_0,
$$
\begin{equation} \label{quadruple-1}
\min_{m_0 < n \leq k_0} n^{\tau} \Delta_n > 2 l \tilde{C}_0, \quad \max_{m_0 \leq n < k_0} n^{\tau} \Delta_n \leq 3 l \tilde{C}_0.
\end{equation}
Now, for some $0 < h < 1$, let $m_1=t(m_0, h)$. Then, for any $m_0 \leq n \leq m_1$, it follows from (\ref{quadruple-1}) and
$$
\theta_n-\theta_{m_0} = (A_n-A_{m_0}) f'(\theta_{m_0})+ R_{m_0, n}, ~~ |R_{m_0, n}| \leq C(m_0^{-\tau} \tilde{C}_0  +  (A_n-A_{m_0})^2 |f'(\theta_{m_0})|)
$$
that
$$
|\Delta_n-\Delta_{m_0}| = O(m_0^{-\tau}) \tilde{C}_0.
$$
Applying (\ref{quadruple-1}), we then deduce that for sufficiently small $h$
\begin{align*}
|n^{\tau} \Delta_{n}-m_0^{\tau} \Delta_{m_0}| & \leq n^{\tau} |\Delta_{n}-\Delta_{m_0}|+(n^{\tau}-m_0^{\tau}) \Delta_{m_0} \\
& \leq O(m_0^{\tau}) O(m_0^{-\tau}) \tilde{C}_0  + o(m_0^{\tau}) 2 l m_0^{-\tau} \tilde{C}_0 ,
\end{align*}
where we have used the fact that
$$
n^{\tau}=O(m_0^{\tau}), ~~ n^{\tau}-m_0^{\tau}=o(m_0^{\tau}).
$$
It then follows that, with $l$ large enough and $h$ small enough, we have
$$
|n^{\tau} \Delta_n-m_0^{\tau} \Delta_{m_0}| \leq l \tilde{C}_0.
$$
In particular, we have
$$
|{(m_0+1)}^{\tau} \Delta_{m_0+1}-m_0^{\tau} \Delta_{m_0}| \leq l \tilde{C}_0 \mbox{ and } |m_1^{\tau} \Delta_{m_1}-m_0^{\tau} \Delta_{m_0}| \leq l \tilde{C}_0,
$$
which further implies that
$$
m_0^{\tau} \Delta_{m_0} \geq l \tilde{C}_0 \mbox{ and } m_1 < k_0,
$$
respectively.

Now, for some $0 < h < 1$, we have
$$
\theta_{m_1}=\theta_{m_0}+ (A_{m_1}-A_{m_0}) f'(\theta_{m_0})+ R_{m_0, m_1},
$$
and
$$
|R_{m_0, m_1}| \leq C(m_0^{-\tau} \tilde{C}_0  +  (A_{m_1}-A_{m_0})^2 |f'(\theta_{m_0})|).
$$
As in the proof of Lemma~\ref{lim-inf}, if $l$ is chosen large enough, then $|f'(\theta_{m_0})|$ will dominate $|R_{m_0, m_1}|$. Again, due to the concavity of $f$, $\Delta_{m_0}$ always has the same sign as $f'(\theta_{m_0})$, then for sufficiently small $h > 0$, we have
$$
|\Delta_{m_1}| \leq  |\Delta_{m_0}|-h/2 |\Delta_{m_0}|.
$$
Then, for $m_0$ sufficiently large such that
$$
m_1^{\tau} < m_0^{\tau}/(1-h/2),
$$
we have
$$
m_1^{\tau} |\Delta_{m_1}| \leq  m_1^{\tau} |\Delta_{m_0}|(1-h/2) < m_0^{\tau} |\Delta_{m_0}| \leq 2 l \tilde{C}_0,
$$
which is a contradiction to (\ref{quadruple-1}).

\end{proof}

\section{Rate of Convergence without Concavity}  \label{without-concavity}

In this section, assuming (\ref{f-convergence}.a) and
\begin{enumerate}[label=(\thesection.\alph*)]
\item with probability $1$, $\theta_n \in Q$ for all $n$, where $Q$ is a compact subset of $\Theta$,
\end{enumerate}
we derive the rate of convergence of our algorithm. Again, for notational convenience only, we assume that $\Theta=\mathbb{R}$.

As one of the main results in real algebraic geometry, the Lojasiewicz inequality~\cite{bi88}, among many other applications, has been widely applied to the convergence analysis of a broad class of dynamical systems. In this section, we will first use the ``function'' version of the Lojasiewicz inequality (Lemma~\ref{Lojasiewicz-2}) to prove that $\{f(\theta_n)\}$ converges almost surely and derive the convergence rate, which can be further used to derive the convergence rate of $\{\theta_n\}$. Then, using the ``variable'' version of the Lojasiewicz inequality (Lemma~\ref{Lojasiewicz-3}), the rate of convergence can be refined. The above-mentioned framework is essentially due to Tadic~\cite{ta10}, however, a comprehensive adaptation to our settings has been done in this section.

Following~\cite{ta10}, we state the ``function'' version of the Lojasiewicz inequality as below.
\begin{lem} \label{Lojasiewicz-2}
For any compact set $Q \subset \Theta$ and real number $z \in f(Q)$, there exist real numbers $\delta_{Q, z} \in (0, 1)$, $\mu_{Q, z} \in (1, 2]$ and $M_{Q, z} \in [1, \infty)$ such that
$$
|f(\theta)-z| \leq M_{Q, z} |f'(\theta)|^{\mu_{Q, z}}
$$
for all $\theta \in Q$ satisfying $|f(\theta)-z| \leq \delta_{Q, z}$.
\end{lem}

From now on, we will set $\hat{f} = \lim_{n \to \infty} f(\theta_n)$ and write $\mu=\mu_{Q, \hat{f}}$. Define
$$
\hat{\Delta}_n \triangleq \hat{f}-f(\theta_n).
$$
We first prove the following lemma.

\begin{lem} \label{larger-than-negative}
There exists a positive integer $l$ such that for all $n$ sufficiently large,
$$
n^{\mu \tau} \hat{\Delta}_n \geq -l \tilde{C}_0^{\mu}.
$$
\end{lem}

\begin{proof}
Suppose, by way of contradiction, that for any $l$, there exists some $n_0$,
\begin{equation} \label{not-too-big}
n_0^{\mu \tau} \hat{\Delta}_{n_0} < -l \tilde{C}_0^{\mu}.
\end{equation}
Then, by Lemma~\ref{EstimationsOfDifferences} (5), we have, for some $0 < h < 1$,
$$
f(\theta_{n_0})-f(\theta_{t(n_0, h)}) \leq -(3/4-3Ch/2)h |f'(\theta_{n_0})|^2+ C \tilde{C}_0^2 n_0^{-2\tau} (1+1/(2h^2)),
$$
which implies for $h$ sufficiently small,
\begin{align*} \label{key-inequality}
\hat{\Delta}_{\theta_{t(n_0, h)}})-\hat{\Delta}_{n_0} & \leq -(3/4-3Ch/2)h |f'(\theta_{n_0})|^2+ C \tilde{C}_0^2 n_0^{-2\tau} (1+1/2h^2) \\
& \leq -h/2 |f'(\theta_{n_0})|^2 + C \tilde{C}_0^2 n_0^{-2\tau} (1+1/(2h^2)).
\end{align*}
Choosing $l$ sufficiently large, then by Lemma~\ref{Lojasiewicz-2} and (\ref{not-too-big}), we deduce that for $n$ large enough,
$$
-h/2 |f'(\theta_{n_0})|^2 + C \tilde{C}_0^2 n_0^{-2\tau} \leq -h/4 |f'(\theta_{n_0})|^2,
$$
and therefore
\begin{equation} \label{dominated-via-Lojasiewicz}
\hat{\Delta}_{t(n_0, h)}-\hat{\Delta}_{n_0} \leq -h/4 |f'(\theta_{n_0})|^2.
\end{equation}
We then have
$$
\hat{\Delta}_{t(n_0, h)} \leq \hat{\Delta}_{n_0} \leq -l \tilde{C}_0^{\mu} n_0^{-\mu \tau} \leq - l \tilde{C}_0^{\mu} t(n_0, h)^{-\mu \tau}.
$$
Henceforth, recursively define
$$
n_{k+1}=t(n_k, h).
$$
It then follows that for any $k$,
$$
\hat{\Delta}_{n_k} \leq \hat{\Delta}_{n_0} \leq -  l \tilde{C}_0^{\mu} n_0^{-\mu \tau} < 0,
$$
which is a contradiction to the fact that almost surely
$$
\lim_{k \to \infty} \hat{\Delta}_{n_k}=0.
$$
\end{proof}

In the remainder of this section, define
$$
\hat{\tau}= \min (\mu \tau, \mu(1-a)/(2-\mu)).
$$
\begin{lem}  \label{bounds-on-liminf}
There exists a positive integer $l$ such that
$$
\liminf_{n \to \infty} n^{\hat{\tau}} \hat{\Delta}_n \leq l \tilde{C}_0^{\mu}
$$
almost surely.
\end{lem}

\begin{proof}
Suppose, by way of contradiction, that for any $l$, we have
\begin{equation}  \label{iMac}
n^{\hat{\tau}} \hat{\Delta}_n \geq l \tilde{C}_0^{\mu},
\end{equation}
for all $n$ sufficiently large. By Lemma~\ref{EstimationsOfDifferences} (5), for any $0 < h < 1$, we have for $n_0$ large enough,
$$
f(\theta_{n_0})-f(\theta_{t(n_0, h)}) \leq -(3/4-3C h/2) h |f'(\theta_{n_0})|^2+ C \tilde{C}_0^2 n_0^{-2\tau} (1+1/(2 h^2)),
$$
which implies that for $h$ sufficiently small
\begin{align*}
\hat{\Delta}_{t(n_0, h)}-\hat{\Delta}_{n_0} & \leq -(3/4-3Ch/2)h |f'(\theta_{n_0})|^2+ C \tilde{C}_0^2 n_0^{-2\tau} (1+1/(2 h^2)) \\
&\leq -h/2 |f'(\theta_{n_0})|^2 +  C \tilde{C}_0^2 n_0^{-2\tau} (1+1/(2 h^2)) .
\end{align*}
Choosing $l$ sufficiently large, then by Lemma~\ref{Lojasiewicz-2} and (\ref{not-too-big}), we deduce that sufficiently large $n$,
$$
-h/2 |f'(\theta_{n_0})|^2 + C \tilde{C}_0^2 n_0^{-2\tau} \leq -h/4 |f'(\theta_{n_0})|^2,
$$
and therefore
\begin{equation} \label{same-as-before}
\hat{\Delta}_{t(n_0, h)}-\hat{\Delta}_{n_0} \leq -h/4 |f'(\theta_{n_0})|^2.
\end{equation}
Now, recursively define
$$
n_{k+1}=t(n_k, h).
$$
An iterated application of (\ref{same-as-before}) yields for some constant $C_1$,
$$
\hat{\Delta}_{n_{k+1}}-\hat{\Delta}_{n_k} \leq -C_1 h \hat{\Delta}_{n_k}^{2/\mu}.
$$

We then have two cases:

{\bf Case $\mathbf{\mu=2}$:} For this case, we have,
$$
\hat{\Delta}_{n_{k+1}} \leq (1 - C_1 h) \hat{\Delta}_{n_k}.
$$
Recursively, we deduce that
$$
\hat{\Delta}_{n_k} \leq \hat{\Delta}_{n_0} (1-C_1 h)^k \leq \hat{\Delta}_{n_0} e^{-C_1 h k},
$$
which implies that for any $k$,
$$
\hat{\Delta}_{n_0} n_k^{\mu \tau} e^{-C_1 h k} \geq n_k^{\mu \tau} \hat{\Delta}_{n_k} \geq l \tilde{C}_0^{\mu}.
$$
This, however, will yield $\tilde{C}_0 \leq 0$ when we take $k$ to $\infty$, which is a contradiction.

{\bf Case $\mathbf{\mu<2}$:} For this case, it follows from
$$
\hat{\Delta}_{n_k}-\hat{\Delta}_{n_{k+1}} \geq C_1 h \hat{\Delta}_{n_k}^{2/\mu}.
$$
that
$$
\int^{\hat{\Delta}_{n_k}}_{\hat{\Delta}_{n_{k+1}}} \frac{1}{u^{2/\mu}} du \geq \int^{\hat{\Delta}_{n_k}}_{\hat{\Delta}_{n_{k+1}}} \frac{1}{\hat{\Delta}_{n_k}^{2/\mu}} du=\frac{\hat{\Delta}_{n_k}-\hat{\Delta}_{n_{k+1}}}{\hat{\Delta}_{n_k}^{2/\mu}} \geq  C h,
$$
which implies that for some positive constant $C_2$
$$
\hat{\Delta}_{n_{k+1}}^{-2/\mu+1}-\hat{\Delta}_{n_k}^{-2/\mu+1} \geq C_2 h.
$$
Recursively, we deduce that
$$
\hat{\Delta}_{n_k}^{-2/\mu+1} \geq \hat{\Delta}_{n_0}^{-2/\mu+1}+C_2 h k,
$$
and furthermore
$$
n_k^{\hat{\tau}(-2+\mu)/\mu} \hat{\Delta}_{n_k}^{(-2+\mu)/\mu} \geq n_k^{\hat{\tau}(-2+\mu)/\mu} \hat{\Delta}_{n_0}^{-2/\mu+1}+ C_2 n_k^{\hat{\tau}(-2+\mu)/\mu} k h.
$$
It then follows from (\ref{iMac}) and (\ref{recursive-n-k}) that
$$
\tilde{C}_0^{-2+\mu} l^{(-2+\mu)/\mu} \geq n_k^{\hat{\tau}(-2+\mu)/\mu} \hat{\Delta}_{n_k}^{(-2+\mu)/\mu} \geq O(kn_k^{\hat{\tau} (-2+\mu)/\mu}) \geq O(k^{\hat{\tau}(-2+\mu)/(-a+1)\mu+1}).
$$
Now, one verifies that this gives us an contradiction if we take $k, l$ to $\infty$, as long as
$$
\hat{\tau} \leq \mu (1-a)/(2-\mu), \mbox{ equivalently } \hat{\tau}(-2+\mu)/(-a+1)\mu+1 \geq 0.
$$

\end{proof}

\begin{lem} \label{bounds-on-limsup}
There exist an integer $l$ such that for all $n$ sufficiently large,
$$
n^{\hat{\tau}} \Delta_n \leq l \tilde{C}_0^{2}.
$$
\end{lem}

\begin{proof}
By way of contradiction, suppose otherwise. Then, by Lemma~\ref{bounds-on-liminf}, for any $l$ and arbitrarily large $N$, we can find $k_0 > m_0 > N$ such that
$$
m_0^{\hat{\tau}}  \Delta_{m_0} \leq 2 l \tilde{C}_0^2, \quad k_0^{\hat{\tau}} \Delta_{k_0}  \geq 3 l \tilde{C}_0^2,
$$
\begin{equation} \label{quadruple-2}
\min_{m_0 < n \leq k_0} n^{\hat{\tau}} \Delta_n > 2 l \tilde{C}_0^2, \quad \max_{m_0 \leq n < k_0} n^{\hat{\tau}} \Delta_n \leq 3 l \tilde{C}_0^2.
\end{equation}

For some $0 < h < 1$, let $m_1=t(m_0, h)$. For any $m_0 \leq n \leq m_1$, as in the proof of Theorem~\ref{larger-than-negative}, we derive
\begin{equation}  \label{one-step-difference}
\hat{\Delta}_n -\hat{\Delta}_{m_0} \leq -h/4 |f'(\theta_{m_0})|^2,
\end{equation}
which, together with Theorem~\ref{larger-than-negative} and (\ref{quadruple-2}), implies that
$$
f'(\theta_{m_0})^2 \leq 4/h \tilde{C}_0^{2} O(m_0^{-\hat{\tau}})+ \tilde{C}_0^{2} O(m_0^{- \mu \tau}),
$$
which, together with (\ref{f-k-n}), further implies that for some $C > 0$,
$$
|\hat{\Delta}_n-\hat{\Delta}_{m_0}| \leq C h |f'(\theta_{m_0})|^2 + C \tilde{C}_0^2 m_0^{-2\tau}
\leq \tilde{C}_0^{2} O(m_0^{-\hat{\tau}})+\tilde{C}_0^{2} O(m_0^{-\mu \tau}) + C \tilde{C}_0^2 m_0^{-2\tau}.
$$
It then follows that for sufficiently small $h$
\begin{align*}
|n^{\hat{\tau}} \hat{\Delta}_n-m_0^{\hat{\tau}} \hat{\Delta}_{m_0}|
& \leq n^{\hat{\tau}} |\hat{\Delta}_n-\hat{\Delta}_{m_0}|+(n^{\hat{\tau}}-m_0^{\hat{\tau}}) \hat{\Delta}_{m_0} \\
& =O(m_0^{\hat{\tau}}) (\hat{\Delta}_n-\hat{\Delta}_{m_0}) + o(m_0^{\hat{\tau}}) \hat{\Delta}_{m_0} \\
& \leq l \tilde{C}_0^2,
\end{align*}
where we have used the fact that
$$
n^{\hat{\tau}}=O(m_0^{\hat{\tau}}), ~~ n^{\hat{\tau}}-m_0^{\hat{\tau}}=o(m_0^{\hat{\tau}}).
$$
In particular, we have
$$
|{(m_0+1)}^{\hat{\tau}} \hat{\Delta}_{m_0+1}-m_0^{\hat{\tau}} \hat{\Delta}_{m_0}| \leq l \tilde{C}_0 \mbox{ and } |m_1^{\hat{\tau}} \hat{\Delta}_{m_1}-m_0^{\hat{\tau}} \hat{\Delta}_{m_0}| \leq l \tilde{C}_0^2,
$$
which further implies that
$$
m_0^{\hat{\tau}} \hat{\Delta}_{m_0} \geq l \tilde{C}_0^2 \mbox{ and } m_1 < k_0,
$$
respectively.

Setting $n=m_1$ and rewriting (\ref{one-step-difference}), we have for some constant $C_1$,
$$
\hat{\Delta}_{m_1}-\hat{\Delta}_{m_0} \leq -C_1 h \hat{\Delta}_{m_0}^{2/\mu}.
$$
We then consider two cases:

{\bf Case $\mathbf{\mu = 2}$:} For this case, we have for some positive constant $C_1$,
$$
\hat{\Delta}_{m_1} \leq (1- C_1 h) \hat{\Delta}_{m_0}.
$$
Then for $m_0$ large enough,
$$
m_1^{\hat{\tau}} \hat{\Delta}_{m_1} \leq (1-C_1 h) m_1^{\hat{\tau}} \hat{\Delta}_{m_0}= (1- C_1 h) m_0^{\hat{\tau}} (1+o(1)) \hat{\Delta}_{m_0} \leq 2 l \tilde{C}_0^2,
$$
which yields a contradiction.

{\bf Case $\mathbf{\mu < 2}$:} For this case, as in the proof of Lemma~\ref{bounds-on-liminf}, we have for some positive constant $C_2$,
$$
\hat{\Delta}_{m_1}^{-2/\mu+1} \geq \hat{\Delta}_{m_0}^{-2/\mu+1}+ C_2 h.
$$
It then follows from (\ref{quadruple-2}) and (\ref{recursive-n-k}) that for $l$ sufficiently large
$$
\hat{\Delta}_{m_1}^{(-2+\mu)/\mu} \geq (2 l \tilde{C}_0^2)^{(-2+\mu)/\mu} m_0^{-a+1} + C_2 h \geq (2 l \tilde{C}_0^2)^{(-2+\mu)/\mu} m_1^{-a+1},
$$
which implies that
$$
m_1^{\hat{\tau}} \hat{\Delta}_{m_1} \leq 2 l \tilde{C}_0^2,
$$
a contradiction.

\end{proof}

The following theorem characterizes the rate of convergence of $\{f(\theta_n)\}$.
\begin{thm} \label{f-convergence-rate}
With probability $1$, we have
$$
|\hat{\Delta}_n| = \tilde{O}(n^{-\hat{\tau}}).
$$
\end{thm}

\begin{proof}
It immediately follows from Lemmas~\ref{bounds-on-liminf} and~\ref{bounds-on-limsup}.
\end{proof}

In the rest of this section,  assuming
\begin{enumerate}[label=(\thesection.\alph*), resume]
\item $\mu \tau \geq (1-a)$,
\end{enumerate}
we prove $\{\theta_n\}$ converges almost surely. Here, let us note that (\thesection.b) can always be satisfied if $a, b, \beta$ are appropriately chosen such that $\tau$ is sufficiently large.

The following theorem characterizes the rate of convergence of $\{\theta_n\}$.
\begin{thm} \label{theta-convergence-rate}
Assume that (\thesection.b). Then, we have
$$
\sup_{k \geq n} |\theta_k - \theta_n| = \tilde{O}(n^{-(\hat{\tau}-(1-a)/2)}).
$$
\end{thm}

\begin{proof}
In this proof, we set
$$
\tau'=(\hat{\tau}+(1-a))/2.
$$
For some $0 < h < 1$, starting from a fixed $n_0$, recursively define
$$
n_{k+1}=t(n_k, h).
$$
Then, to prove the theorem, it suffices to prove that
\begin{equation} \label{quantized}
\sup_{k \geq m} |\theta_{n_k} - \theta_{n_m}| = \tilde{O}(n_m^{-({\hat{\tau}}+(1-a)/2)}).
\end{equation}

Now, applying Lemma~\ref{EstimationsOfDifferences} (7), we deduce that for some $C > 0$
$$
|\theta_{n_{i+1}}-\theta_{n_i}| \leq C n_i^{\tau'} (f(\theta_{n_{i+1}})-f(\theta_{n_i}))+C \tilde{C}_0^2 n_i^{-\tau'}.
$$
It then follows that for any $m \leq k$,
\begin{align*}
|\theta_{n_k}-\theta_{n_m}| & \leq \sum_{i=m}^{k-1} |\theta_{n_{i+1}}-\theta_{n_i}| \\
&\leq C \tilde{C}_0^2 \sum_{i=m}^{k-1} n_i^{-\tau'}+ C \sum_{i=m}^{k-1} (u(\theta_{n_i})-u(\theta_{n_{i+1}})) n_i^{\tau'}\\
&\leq C \tilde{C}_0^2 \sum_{i=m}^{k-1} n_i^{-\tau'} + C \sum_{i=m+1}^k (n_i^{\tau'}-n_{i-1}^{\tau'}) |u(\theta_{n_i})|+ C n_m^{\tau'} |u(\theta_{n_m})|+ C n_k^{\tau'} |u(\theta_{n_k})|.
\end{align*}

Applying (\ref{recursive-n-k}), we deduce that
\begin{align*}
\sum_{i=m}^{k-1} n_i^{-\tau'}& =\sum_{i=m}^{k-1} O(i^{-\tau'/(1-a)}) = O(m^{-\tau'/(1-a)+1}), \\
\sum_{i=m+1}^k (n_i^{\tau'}-n_{i-1}^{\tau'}) |u(\theta_{n_i})| & = \sum_{i=m}^k O((i-1)^{y/(1-a)-1} i^{-{\hat{\tau}}/(1-a)}) = O(m^{y/(1-a)-{\hat{\tau}}/(1-a)}), \\
n_m^{\tau'} |u(\theta_{n_m})| &=  O(m^{\tau'/(1-a)} m^{-{\hat{\tau}}/(1-a)})=O(m^{(\tau'-{\hat{\tau}})/(1-a)}),\\
n_k^{\tau'} |u(\theta_{n_k})| & = O(k^{\tau'/(1-a)} k^{-{\hat{\tau}}/(1-a)})=O(k^{(\tau'-{\hat{\tau}})/(1-a)}).
\end{align*}
We then immediately conclude that
$$
|\theta_{n_k}-\theta_{n_m}| = O(n_m^{(\tau'-{\hat{\tau}})}),
$$
which immediately implies (\ref{quantized}).

\end{proof}

The following ``variable'' version of the Lojasiewicz inequality will be used to refine the rates of convergence of $\{\theta_n\}$ and $\{f(\theta_n)\}$.
\begin{lem} \label{Lojasiewicz-3}
For each $\theta \in \Theta$, there exist real numbers $\delta_{\theta} \in (0, 1)$,
$\mu_{\theta} \in (1, 2]$, $M_{\theta} \in [1, \infty)$ such that
$$
|f(\theta')-f(\theta)| \leq M_{\theta} \|f'(\theta')\|^{\mu_{\theta}}
$$
for all $\theta' \in \Theta$ satisfying $\|\theta'-\theta\| \leq \delta_{\theta}$.
\end{lem}

Theorem~\ref{theta-convergence-rate} implies that with probability $1$, $\{\theta_n\}$ converges. From now on, let $\hat{\theta}=\lim_{n \to \infty} \theta_n$ and set $\mu=\mu_{\hat{\theta}}$. Then, with this redefined $\mu$, going through exactly the same arguments as in the proof of Theorems~\ref{f-convergence-rate} and~\ref{theta-convergence-rate}, we have the following two theorems.
\begin{thm}
For the above redefined $\mu$, Theorems~\ref{f-convergence-rate} holds.
\end{thm}

\begin{thm}
For the above redefined $\mu$, assume (\thesection.b). Then, we have
$$
|\theta_n-\hat{\theta}| = \tilde{O}(n^{-(\hat{\tau}-(1-a)/2)}).
$$
\end{thm}

\section{Capacity Achieving Distribution of a Special Class of Channels} \label{Memoryless-Channels}

In this section, we restrict our attention to a special class of input-restricted finite-state channels with certain parameterization and we prove that for such channels operated at high SNR regime, the capacity will only be achieved at the interior of the parameter space and our algorithm converges almost surely.

More specifically, recalling $X, Y$ denote the input, output processes of the channel over finite alphabets $\mathcal{X}$, $\mathcal{Y}$, respectively, we consider a class of parameterized memoryless channels such that
\begin{enumerate} [label=(\thesection.\alph*)]
\item the channel only has one state; in other words, at any time slot, the channel is characterized by the conditional probability $p(y|x)$.

\item for some mixing finite-type constraint $F \subset \mathcal{X}^2$, $X \in \Pi_F$.

\item the channel is parameterized by $\eps \ge 0$ such that for each $x$ and $y$, $p(y|x)(\eps)$ is an analytic function of $\eps \ge 0$, which is not identically $0$.

\item there is a one-to-one (not necessarily onto) mapping $\Phi: \mathcal{X} \to \mathcal{Y}$, such that for any $x \in \mathcal{X}$, $p(\Phi(x)|x)(0) = 1$.
\item $X$ is parameterized as in~\cite{pa04}, that is,
$$
\theta=(p(X_{1}=w_1, X_{2}=w_2): (w_1, w_2) \not \in F).
$$
\end{enumerate}

Under the above assumptions, $\eps$ can be regarded as a parameter that quantifies noise, and $\Phi(x)$ is the noiseless output corresponding to input $x$. The regime of ``small $\eps$'' corresponds to high SNR. Note that the output process $Y = Y(X,\eps)$ depends on the input process $X$ and the parameter value $\eps$; we will often suppress
the notational dependence on $\eps$ or $X$, when it is clear from the context. Prominent examples of such families include input-restricted versions of the binary symmetric channel with crossover probability $\eps$, denoted by BSC($\eps)$,, and the binary erasure channel with erasure rate $\eps$, denoted by BEC($\eps)$.

\textbf{General SNR regime.} By using an asymptotic formula of $I(X; Y)$, we show that for the above-mentioned channels, the capacity achieving $X$ must be primitive.

Assume that $X$ has period $e$ with period classes $D_1, D_2, \ldots, D_e$. Then, by the classical Perron-Frobenius theory, after necessary reindexing, its transition probability matrix $\Pi$ can be written as
\begin{equation}   \label{B-matrix}
\bordermatrix{ & D_1 & D_2 & D_3 & \cdots & D_e \cr
D_1 & 0 & B_1 & 0 & \cdots & 0 \cr
D_2 & 0 & 0 & B_2 & \cdots & 0 \cr
\vdots & \vdots & \vdots & \vdots & \ddots & \vdots \cr
D_{e-1} & 0 & 0 & 0 & \cdots & B_{e-1} \cr
D_e & B_e & 0 & 0 & \cdots & 0 \cr},
\end{equation}
where we used the period classes to index the sub-blocks. In the following, let $\mathbf{B}$ denote the set of all entry indices of $\Pi$ corresponding to some $B_k$, that is,
$$
\mathbf{B}=\{(i, j): i \in D_k, j \in D_{k+1}, \mbox{ for } k=1, \cdots, e-1 \} \cup \{(i, j): i \in D_e, j \in D_1\}.
$$
Now, consider an analytic perturbation $\Pi(\delta)$ of $\Pi$, $\delta \geq 0$, where
\begin{enumerate}[label=(\thesection.\alph*), resume]
\item $\Pi(0)=\Pi$;
\item for some $(i, j) \in \mathbf{B}$, $\Pi_{ij}(\delta)$ is not identically $0$;
\item for any $\delta \geq 0$, $\Pi(\delta)$ is still a stochastic matrix.
\end{enumerate}
In other words, some non-$B$-entries in $\Pi$ are analytically perturbed; as a result, $Y$ is perturbed from $Y(0)$ to $Y(\delta)$. The following theorem describes the asymptotic behavior of $H(Y)$ under such a perturbation.

\begin{thm} \label{has-to-be-primitive}
Under the aboved-mentioned perturbation as in (\thesection.f)-(\thesection.f), there exist $C_1, C_2 > 0$ such that
$$
C_1 \delta \log 1/{\delta} \leq H(Y(\delta)) - H(Y(0)) \leq C_2 \delta^{1/2}.
$$
\end{thm}

\begin{proof}
The proof is postponed to Appendix~\ref{proof-of-has-to-be-primitive}.
\end{proof}

\begin{rem} \label{remark-after-has-to-be-primitive}
It follows from Condition (\thesection.a) that $H(Y|X)$ is linear with respect to $\vec{p}$. Theorem~\ref{has-to-be-primitive}, together with this fact, implies that there exist $C_1, C_2$ such that
$$
C_1 \delta \log 1/{\delta} \leq I(X(\delta); Y(\delta))-I(X(0); Y(0)) \leq C_2 \delta^{1/2},
$$
which implies that, for any irreducible but not primitive $X$, any perturbation of $X$ as in (\thesection.f)-(\thesection.h) will strictly increase the mutual information. So, we conclude that the capacity achieving $X$ must be primitive, and thus Condition (\ref{channel-model}.a) holds.
\end{rem}

\textbf{High SNR regime.} At the high SNR regime, that is, when $\eps$ is close to $0$, it has been established in~\cite{hm09b} that there exists $\hat{\epsilon} > 0$ such that
\begin{enumerate} [label=(\thesection.\alph*), resume]
\item $I(X; Y)$, when restricted on $X \in \Pi_{F, \hat{\epsilon}}$, is strictly concave with respect to $\theta \in \Theta$.
\item the capacity of the channel can be uniquely achieved within $\Pi_{F, \hat{\epsilon}}$.
\end{enumerate}
As a consequence, we have the following theorem.
\begin{thm} \label{high-SNR-case}
For the channel as in (\thesection.a)-(\thesection.d) operating at the high SNR regime and sufficiently small $\hat{\epsilon}$, under the iteration in (\ref{theta-iteration}), $\{\theta_n\}$ converges to the capacity achieving distribution with probability $1$.
\end{thm}

\begin{proof}
Note that Condition (\thesection.a) and Theorem~\ref{convergence-theorem} imply Conditions (\ref{with-concavity}.a) and (\ref{with-concavity}.b); and Condition(\thesection.b) implies that the global maximum $\theta^*$ indeed corresponds to the capacity achieving distribution. The theorem then immediately follows.
\end{proof}

\begin{exmp} \label{not-concave-example}
Consider a binary symmetric channel with crossover probability $\eps > 0$. Let $X$ be a binary input Markov chain with the transition probability matrix
\begin{equation}
\label{matrixP} \left[
\begin{array}{cc}
1-\pi & \pi \\
1& 0
\end{array}
\right],
\end{equation}
where $0 \leq \pi \leq 1$. Apparently, $X$ is supported on the so-called $(1,\infty)$-RLL constraint~\cite{lm95}, which simply means that the string ``$11$'' is forbidden. Let $Y$ denote the corresponding output process. Assume that $X$ is parameterized by $\theta=(p(00), p(01), p(10))$, where $p(10)=1$ is in fact a constant. It can be checked that Conditions (\thesection.a)-(\thesection.d) are all satisfied, so when $\eps$ is sufficiently small, Conditions (\thesection.i)-(\thesection.j) are satisfied and thus Theorem~\ref{high-SNR-case} holds.

On the other hand, it has been shown that for the output process $Y$, as $\eps \to 0$,
\begin{equation} \label{et2}
H(Y)=H(X) + \frac{\pi(2-\pi)}{1+\pi} \eps \log (1/\eps) +O(\eps),
\end{equation}
where the $O(\eps)$-term is analytic with respect to $p$ (see Theorem $2.18$ of~\cite{hm09b}). It then follows that
$$
H(X|Y)=H(X)+H(Y|X)-H(Y)=H(\eps)-\frac{\pi(2-\pi)}{1+\pi} \eps \log (1/\eps) +O(\eps),
$$
where $H(\eps)=\eps \log 1/\eps+(1-\eps) \log 1/(1-\eps)$. One can readily verify that $-\pi(2-\pi)/(1+\pi)$ is strictly convex with respect to $\theta$, which implies the strict convexity (rather than concavity) of $H(X|Y)$ when $\eps$ is small enough. So, the concavity conjecture in~\cite{pa04} is not true in general, and thus the conditions guaranteeing the convergence of the GBAA are not satisfied.
\end{exmp}

\section*{Appendices}\appendix

\section{Proof of Theorem~\ref{has-to-be-primitive}}  \label{proof-of-has-to-be-primitive}

First of all, we define
$$
Z(\delta)=Z(X_1^n(\delta))=\begin{cases}
0 & (X_i(\delta), X_{i+1}(\delta)) \in \mathbf{B} \mbox{ for all $i \in \{1, \cdots, n-1\}$}\\
1 & (X_i(\delta), X_{i+1}(\delta)) \not \in \mathbf{B} \mbox{ for exactly one $i \in \{1, \cdots, n-1\}$} \\
2 & (X_i(\delta), X_{i+1}(\delta)) \not \in \mathbf{B} \mbox{ for more than one $i \in \{1, \cdots, n-1\}$}
\end{cases}.
$$
Next, applying the Birch bound~\cite{bi62}, we derive the following key inequality for this proof:
\begin{equation}  \label{the-birch-bound}
\hspace{-1cm} \frac{H(Y_{m+1}^n(\delta)|Y_1^m(\delta), X_0(\delta), Z(\delta))}{n-m} \leq H(Y) \leq \frac{H(Y_{1}^n|X_0(\delta), Z(\delta))}{n}+\frac{H(Z(\delta))}{n}+\frac{H(X_0(\delta))}{n},
\end{equation}
for any $m \leq n$.

\textbf{The lower bound part.} We first prove that there exists $C_1 > 0$ such that
$$
H(Y(\delta)) \geq H(Y(0)) + C_1 \delta \log 1/{\delta},
$$
which immediately implies the lower bound part of the theorem.
In this part, we set
\begin{equation}  \label{set-m-n-1}
n= \sqrt{\log \delta} \mbox{ and } m = n/2.
\end{equation}
By definition, we have
\begin{align*}
\hspace{-2cm}  H(Y_{m+1}^n(\delta)|Y_1^m(\delta), X_0(\delta), Z(\delta))/(n-m) & = \sum_{x_0} p^{\delta}(x_0, Z=0) H(Y_{m+1}^n(\delta)|Y_1^m(\delta), X_0(\delta), Z(\delta)=0)/(n-m) \\
&+\sum_{x_0} p^{\delta}(x_0, Z=1) H(Y_{m+1}^n(\delta)|Y_1^m(\delta), X_0(\delta), Z(\delta)=1)/(n-m) \\
&+\sum_{x_0} p^{\delta}(x_0, Z=2) H(Y_{m+1}^n(\delta)|Y_1^m(\delta), X_0(\delta), Z(\delta)=2)/(n-m). \\
& \triangleq T_1+T_2+T_3
\end{align*}
where $p^{\delta}(x_0, Z=0)$ means $P(X_0(\delta)=x_0, Z(\delta)=0)$.

We next give estimates for the each of three terms defined as above.

For $T_3$, notice that $n \delta < 1$ for sufficiently small $\delta$ and then
$$
\sum_{x_0} p^{\delta}(x_0, Z=2) \leq n^2 (C_0 \delta)^2+ n^3 (C_0 \delta)^3+\cdots \leq \frac{C_0^2}{1-n C_0 \delta} n^2 \delta^2,
$$
for some $C_0 > 0$. It then follows that
\begin{align} \label{T-3}
\nonumber T_3 & = \sum_{x_0} p^{\delta}(x_0, Z=2) H(Y_{m+1}^n(\delta)|Y_1^m(\delta), X_0(\delta), Z(\delta)=2)/(n-m) \\
\nonumber & \leq \sum_{x_0} p^{\delta}(x_0, Z=2) H(Y_{m+1}^n(\delta))/(n-m) \\
\nonumber & \leq \sum_{x_0} p^{\delta}(x_0, Z=2) \log |\mathcal{Y}| \\
& = O(n^2 \delta^2).
\end{align}

For $T_2$, one verifies that for any $x_0$, there exist constants $C_1, C_2 > 0$, $0 < \lambda_1 < \lambda_2 < 1$ such that
$$
C_1 n \delta \lambda_1^n \leq p^{\delta}(y_1^n|x_0, Z=1) \leq C_2 n \delta \lambda_2^n.
$$
Similarly, for any $x_0$, there exist $C_3, C_4 > 0$, and the same $0 < \lambda_1 < \lambda_2 < 1$ as above such that
$$
C_3 m \lambda_1^m \leq p^{\delta}(y_1^m|x_0, Z=1) \leq C_4 m \lambda_2^m.
$$
It then follows that for any $x_0$,
$$
C_5 \delta \lambda_2^n/\lambda_1^m \leq p^{\delta}(y_{m+1}^n|y_1^m, X_0, Z=1) \leq C_6 \delta \lambda_2^n/\lambda_1^m,
$$
which, together with (\ref{set-m-n-1}), implies that
$$
H(Y_{m+1}^n(\delta)|Y_1^m(\delta), X_0(\delta), Z(\delta)=1)= \hat{O} (\log 1/\delta)+ O(n \log \lambda_2)+ O(m \log \lambda_1).
$$
This, together with the fact
$$
p(x_0, Z=1)=\hat{O}(n \delta),
$$
implies that
\begin{equation}  \label{T-2}
T_2=\hat{O} (\delta \log 1/\delta)+ O(n \delta \log \lambda_2)+ O(m \delta \log \lambda_1).
\end{equation}

For $T_1$, notice that it can be rewritten as
$$
T_1=\sum p^{\delta}(y_1^n, x_0, Z=0) \log p^{\delta}(y_{m+1}^n|y_1^m, x_0, Z=0)/(n-m).
$$
One then verifies that
$$
\left| p^{\delta}(y_1^n, x_0, Z=0)-p^0(y_1^n, x_0, Z=0)|_{\delta=0} \right|= O(n \delta) p^0(y_1^n, x_0, Z=0),
$$
which implies that
$$
\hspace{-1cm}  \left|\frac{\sum p^0(y_1^n, x_0, Z=0) \log p^{\delta}(y_{m+1}^n|y_1^m, x_0, Z=0)}{n-m} \right.
\left. -\frac{\sum p^{\delta}(y_1^n, x_0, Z=0) \log p^{\delta}(y_{m+1}^n|y_1^m, x_0, Z=0)}{n-m} \right|=O(n \delta).
$$
When fixing $x_0$ and assuming $Z=0$, the analyticity argument in~\cite{gm05} can be used to prove that
$$
\sum p^0(y_1^n, x_0, Z=0) \log p^{\delta}(y_{m+1}^n|y_1^m, x_0, Z=0)/(n-m)
$$
exponentially converges to an analytic function of $\delta$. It then follows that
for some $0 < \rho < 1$
\begin{equation} \label{T-1}
T_1=H(Y(0))+O(\rho^m)+O(\delta).
\end{equation}
Combining (\ref{T-3}), (\ref{T-2}) and (\ref{T-1}), we then have
$$
H(Y_{m+1}^n(\delta)|Y_1^{m}, X_0(\delta), Z(\delta))/(n-m)=H(Y(0))+\hat{O}(\delta \log 1/\delta).
$$

\textbf{The upper bound part.} We then prove that there exists $C_2 > 0$,
$$
H(Y(\delta)) \leq H(Y(0)) + C_2 \delta^{1/2},
$$
which immediately implies the upper bound part of the theorem. For this part, setting
\begin{equation}  \label{set-m-n-2}
n=\delta^{-1/2} \mbox{ and } m=0.
\end{equation}
Using a parallel argument as in the lower bound part, we can still derive (\ref{T_3}), (\ref{T_2}) and (\ref{T_1}) and then
$$
H(Y_1^n|X_0(\delta), Z(\delta))/n=H(Y)_{\delta=0, 0}+O(\delta^{1/2}).
$$

It can verified that
$$
p(Z(\delta)=1)= O(n \delta), \quad p(Z(\delta)=2) = O(n^2 \delta^2),
$$
which, together with the straightforward fact $H(X_0(\delta))/n=O\left(\frac{1}{n}\right)$, implies that
$$
H(Z(\delta))= - \sum_{i=0}^2 p(Z(\delta)=i) \log p(Z(\delta)=i) = O(n \delta \log \delta)+O(n \delta \log n),
$$
and consequently
$$
\frac{H(Z(\delta))}{n} = O(\delta \log \delta)+O(\delta \log n).
$$
The upper bound part then follows from all the above estimates and (\ref{the-birch-bound}).

\end{document}